%% file: main.tex
\documentclass[12pt]{article}
\usepackage{amsfonts,amsmath,amssymb}
\usepackage{enumerate, graphicx}
\usepackage{xspace}
\usepackage{boxedminipage}
\usepackage{url}
\usepackage{algo}
\usepackage{enumerate, paralist}
 \usepackage[usenames,dvipsnames]{pstricks}
 \usepackage{pst-grad} 
 \usepackage{pst-plot} 
 \usepackage{bbm} 
 \usepackage{float} 
\usepackage{hyperref} 
\usepackage{multirow} 
\usepackage{algorithm} 
\usepackage[bb=boondox]{mathalfa} 
\usepackage{amsmath,amssymb,amsthm}
\usepackage{mathtools} 
\usepackage{thmtools, thm-restate}
\usepackage{soul} 
\usepackage{} 
\usepackage{todonotes}
\usepackage{diagbox}

\usepackage{algorithm} 
\usepackage[noend]{algpseudocode} 

\hypersetup{
    colorlinks,
    linkcolor={Red},
    citecolor={ForestGreen},
    urlcolor={Violet}
}

\usepackage[margin=1in]{geometry}
\usepackage{setspace}
\usepackage[english]{babel}
\usepackage[utf8]{inputenc}
\usepackage{graphicx}

\usepackage{boxedminipage}
\usepackage{alltt}

\newtheorem*{theorem*}{Theorem}

\newtheorem{lemma}{Lemma}[section]
\newtheorem{theorem}[lemma]{Theorem}

\newtheorem{corollary}[lemma]{Corollary}

\newtheorem{defn}[lemma]{Definition}

\theoremstyle{definition}
\newtheorem{remark}[lemma]{Remark}

 \newenvironment{proofof}[1]{\smallskip\noindent{\bf Proof of #1:}}%

\usepackage{mathtools}
\usepackage{enumerate}
\usepackage{thmtools}
\usepackage{thm-restate}

\def\E{\ensuremath{\mathrm{\mathbf{E}}}}
\def\Var{\ensuremath{\mathrm{\mathbf{Var}}}}
\def\Bin{\ensuremath{\mathrm{\mathbf{Bin}}}}
\def\Poi{\ensuremath{\mathrm{\mathbf{Poi}}}}
\def\OPT{\ensuremath{\mathrm{\mathbf{OPT}}}}
\def\VV{\ensuremath{\mathcal{V}}}
\def\AA{\ensuremath{\mathcal{A}}}
\def\HH{\ensuremath{\mathcal{H}}}
\def\XX{\ensuremath{\mathcal{X}}}
\def\PP{\ensuremath{\mathcal{P}}}
\def\NN{\ensuremath{\mathcal{N}}}
\def\LP[#1]{\ensuremath{\mathrm{\mathbf{LP#1}}}}
\def\OP[#1]{\ensuremath{\mathrm{\mathbf{OP#1}}}}
\def\Big{\ensuremath{\mathsf{Big}}}
\def\Mon{\ensuremath{\mathsf{Mon}}}
\def\pmax{\ensuremath{p_{\textrm{max}}}}

\def\accept{{\fontfamily{cmss}\selectfont accept}\xspace}
\def\reject{{\fontfamily{cmss}\selectfont reject}\xspace}

\def\ceil#1{\lceil {#1} \rceil}

\title{Towards Testing Monotonicity of Distributions Over General Posets}

\author{
Maryam Aliakbarpour
\thanks{MA is supported by funds from the MIT-IBM Watson AI Lab (Agreement No. W1771646),  the NSF grants IIS-1741137, and CCF-1733808.}\\
CSAIL, MIT\\
\texttt{maryama@mit.edu}
\and
Themis Gouleakis
\thanks{TG is supported by the NSF grants CCF-1740751, CCF-1650733, CCF-1733808, and IIS-1741137. Part of this work was done while TG was a postdoctoral researcher at USC supported by Ilias Diakonikolas' USC startup grant.
}\\
Max Planck Institute\\
\texttt{tgouleak@mpi-inf.mpg.de}
\and 
John Peebles
\thanks{JP is supported by the NSF grants CCF-1565235, CCF-1650733, CCF-1733808, and IIS-1741137. 
}
\\
CSAIL, MIT\\
\texttt{jpeebles@mit.edu}
\and
Ronitt Rubinfeld\thanks{RR is supported by by funds from the MIT-IBM Watson AI Lab (Agreement No. W1771646),
the NSF grants CCF-1650733, CCF-1733808, IIS-1741137 and CCF-1740751.
}\\
CSAIL, MIT, TAU \\
\texttt{ronitt@csail.mit.edu}
\and 
Anak Yodpinyanee
\thanks{AY is supported by the NSF grants CCF-1650733, CCF-1733808, IIS-1741137 and the DPST scholarship, Royal Thai Government. This work was completed while AY was at CSAIL, MIT.}
\\
CSAIL, MIT\\
\texttt{anak@mit.edu}
}

\begin{document}
\maketitle

\vspace{-0.4cm}
\begin{abstract}{In this work, we consider the sample complexity required for testing the monotonicity of distributions over partial orders. 
A distribution $p$ over a poset is {\em monotone} if, for any pair of domain elements $x$ and $y$ such that $x \preceq y$, $p(x) \leq p(y)$.

To understand the sample complexity of this problem, we introduce a new property called \emph{bigness} over a finite domain, where the distribution is $T$-big if the minimum probability for any domain element is at least $T$.
We establish a lower bound of $\Omega(n/\log n)$ for testing bigness of distributions on  domains of size $n$. We then build on 
these lower bounds to give $\Omega(n/\log{n})$ lower bounds for 
testing monotonicity over a matching poset of size $n$ and significantly improved
lower bounds over
the hypercube poset.
 
We give sublinear sample complexity bounds for testing bigness and for testing
monotonicity over the matching poset.  
 We then give a number of tools for analyzing upper bounds on the sample complexity of
 the monotonicity testing problem. 
 
 \textbf{Keywords:} Property Testing; Monotone Distributions; Partially Ordered Sets;
} \end{abstract}


\section{Introduction}

We consider the problem of testing whether a distribution is monotone: an essential property that captures many observed phenomena of real-world probability distributions. For instance, monotone distributions over \emph{totally ordered sets}
might be used to describe distributions on diseases for which the probability of being affected by the disease increases with age.  More generally, an important class of distributions are characterized by being monotone
over a \emph{partially ordered set} (poset). For these distributions, if a domain element $u$ lower bounds $v$ in the partial ordering (denoted $u \preceq v$), then $p(u) \leq p(v)$ (whereas if $u$ and $v$ are unrelated in the poset, then $p$ needs not satisfy any
particular requirement on the relative probabilities of $u$ and $v$).   Such distributions might include distributions on diseases for which the probability of being affected increases by some combination
of several risk factors. 
Many commonly studied distributions, e.g. exponential distributions or multivariate exponential distributions, are or can be approximated by piecewise monotone functions.
As monotone distributions are a fundamental class of distributions, the problem of testing whether a distribution is monotone is a key building block for distribution
testing algorithms.  

Given an unknown distribution, over a poset domain, 
the goal is to distinguish whether the distribution is 
monotone 
or far from any  monotone distribution, using as few samples as possible.
This problem has been considered in the literature:
the problem of testing whether a distribution is monotone was first considered in the work of \cite{BatuKR04},
where testing the monotonicity of distributions over totally  ordered domains and 
partially ordered domains that corresponded to two-dimensional grids  were considered.
The work of \cite{BhattacharyyaFRV10} introduced the study of testing the monotonicity of
distributions over general partially ordered domains, and in particular, considered
the Boolean hypercube  ($\{0, 1\}^d$).  Several other works considered these questions
\cite{DaskalakisDSV13,AcharyaDK15,CanonneDGR18} under various different domains and achieved improved
sample complexity bounds.

The sample complexity of the testing problem varies greatly with the structure of the poset: On the one hand, for domains of size $n$ that are total orders, $\Theta(\sqrt{n})$ samples suffice for distinguishing monotone distributions, from those that are $\epsilon$-far in total variation distance from any monotone distribution \cite{BatuKR04,AcharyaDK15,CanonneDGR18}. On the other hand,
testing distributions defined over the matching poset requires nearly linear in $n$, specifically $\Omega(n^{1 - o(1)})$, samples \cite{BhattacharyyaFRV10}.   Furthermore,
for a large class of familiar posets, such as the Boolean hypercubes, little is understood
about the sample complexity of the testing problem.

\paragraph{Our results and approaches:} 
We first define a new property called the {\em bigness} property, which we use as our main building block for establishing sample complexity lower bounds for monotonicity testing. A distribution is {\em $T$-big} if every domain element is assigned probability mass at least $T$. 

Though the bigness property is a {\em symmetric} property (i.e.,  
permuting the labels of the elements does not change 
whether the distribution has the property or not), 
we use lower bounds for testing the bigness property in order to prove 
lower bounds on testing monotonicity, which is not a symmetric property.  In addition, the bigness property is a natural property, and thus of interest in its own right.


We show that the sample complexity 
of the bigness testing problem is $\Theta(n/\log n)$ when $T=\Theta(1/n)$. The upper bound follows from applying the algorithm of \cite{ValiantV17} that learns the underlying distribution up to a permutation of the domain elements.
Our lower bound approach is inspired by the framework of \cite{wu2015chebyshev}, used to lower
bound the number of samples needed to estimate support sizes. Our lower bound is established by showing that the distribution of samples, one generated from $T$-big distributions ($p$'s) and the other generated from distributions that are $\epsilon$-far from $T$-big ($p'$'s), are statistically close. In contrast with the standard lower bound framework, $p$ and $p'$ are not picked from two sets of distributions.
Instead, the distribution $p$ (resp.~$p'$) is constructed by having each domain element $i$ choose its probability $p(i)$, in an i.i.d.~fashion, from the distribution $P_V$ (resp.~$P_{V'}$) over possible probabilities in $[0,1]$. 
To design $P_V$ and $P_{V'}$, we introduce a new optimization problem that maximizes $\epsilon$ while keeping the distribution of samples statistically close. This constraint is established via the \emph{moments matching} technique, which
allows us to show that the distributions are
indistinguishable with $o(n/\log{n})$ samples, but also plays a crucial role  in many other settings  \cite{RaskhodnikovaRSS09,Valiant08,BhattacharyyaFRV10,ValiantV16,ValiantV17,wu2015chebyshev,WuY16Entropy}.

By reducing from the bigness testing problem, we next give a  lower bound 
of $\Omega(n/\log{n})$
on the sample complexity of the monotonicity testing problem over the matching poset, 
improving on the $\Omega\left(n/2^{\Theta(\sqrt {\log n})}\right)$ lower bound in \cite{BhattacharyyaFRV10}.  
In addition to improving the sample complexity lower bound, one particularly useful byproduct of our approach is that the maximum probability of an element in the constructed lower bound distribution families can be made small, which assists us in proving lower bounds for other posets in the following.

Finally, we leverage the lower bound for the monotonicity testing problem over the matching poset to prove a lower bound of $N^{1-\delta}$ for $\delta = \Theta(\sqrt{\epsilon}) +o(1)$ for monotonicity testing over the Boolean hypercube of size $N=2^d$, greatly improving upon the standard ``Birthday Paradox'' lower bound of $\Omega(\sqrt{N})$. 
Our reduction follows from finding a large embedding of the matching poset in the hypercube, and its efficiency follows from the previously mentioned
upper bound on the maximum element probability from the bigness lower bound construction above.

We then give a number of new tools for analyzing upper bounds on the sample complexity of the monotonicity testing problem:
 \begin{enumerate}
 \item We prove that the distance of a distribution to monotonicity can be characterized approximately as the weight of a \emph{maximum weighted matching} in the \emph{transitive closure} of the poset, where the weight of the edge $(u,v)$ is the amount of violation from being monotone: $\max(0, p(u) - p(v))$. This characterization gives a structural result about distributions that are $\epsilon$-far from monotone. 
Moreover, this results extends the work of \cite{FischerLNRRS02} to non-boolean valued functions. The work of \cite{FischerLNRRS02} shows that the distance of a boolean function $f$ to monotonicity is related to the number of ``violating edges" in the transitive closure of the underlying poset.  
 
 \item Via the characterization above, we show that the monotonicity testing problem over \emph{bipartite} posets (where all edges are directed in the same direction) captures the monotonicity testing problem in its full generality. That is, we give a reduction from monotonicity testing over any poset to monotonicity testing over a bipartite poset.  Our reduction preserves the number of vertices and the distance parameter up to a constant multiplicative factor. 
As in the previous, this result extends the work of \cite{FischerLNRRS02} to non-boolean valued functions.

\item Leveraging the learning algorithms for {\em symmetric} distributions in \cite{ValiantV17}, we propose algorithms with sample complexity $O(n/(\epsilon^2\log n))$ for testing bigness of a distribution, and for testing monotonicity on matching posets. The proof of our latter result requires certain subtle details: (1) an additional reduction that allows us to scale our distribution for ``each side'' of the matching, in order to generate sufficient samples from each side, as required by the algorithm of \cite{ValiantV17}, and (2) technical lemmas establishing bounds between the total variation distance and the distance notion in \cite{ValiantV17}, under the scaling mentioned earlier.

\item We give a reduction from monotonicity testing on a bipartite poset, to monotonicity testing on the matching (for which the testing algorithm is constructed above). This reduction gives an algorithm for monotonicity testing on any bipartite poset (which is the most general problem, as argued earlier), in which the overhead in the sample complexity depends only on the maximum degree of the bipartite graph.
 \item  We give another upper bound for testing monotonicity on bipartite
 posets: $O((\log M)/\epsilon^2)$ where $M$ is the number of ``endpoint sets'' of all possible matchings contained in the given bipartite graph (or equivalently, the number of induced subgraphs that admit a perfect matching over their respective vertex sets). Note that for the matching poset, $M = 2^n$ yields an $O(n/\epsilon^2)$ upper bound, and therefore for matching posets
 our previous algorithm is preferable.   However, this bound yields an upper bound of $O(n/\epsilon^2)$   for {\em all} posets, and could potentially be even smaller for certain classes of graphs, such as collections of large stars.
 \item Finally, we give an upper bound of $O(\frac{n^{2/3}}{\epsilon}+\frac{1}{\epsilon^2})$ samples for monotonicity testing on bipartite posets, under the guarantee that the distribution being tested is a uniform distribution on some subset of known size of the domain. This special case is of interest in that it relates to the well studied problem of testing monotonicity of \emph{Boolean functions}, in a somewhat different setting where instead of getting query access to the function, we are given uniform ``positive'' samples of domain elements $x$ for which $f(x) = 1$.
 \end{enumerate}

\paragraph{Other related work}
Batu, Kumar, and Rubinfeld \cite{BatuKR04} initiated the study of testing monotonicity of distributions.  For the case where the domain is totally ordered, the sample complexity is known to be ${\Theta}(\sqrt n)$ \cite{BatuKR04, AcharyaDK15, CanonneDGR18}.
Several works have considered distributions over higher dimensional domains. In \cite{BatuKR04,BhattacharyyaFRV10}, it is shown that  testing monotonicity of a distribution on the  two dimensional grid $[m]\times[m]$  (here $N=m^2$) can be performed using $\widetilde{O}(N^{3/4})$ samples. 
For higher dimensional grids $[m]^d$ (where $N=m^d$), Bhattacharyya et al.  provided an algorithm that uses $\widetilde{O}(m^{d - 1/2}) = \widetilde{O}(N/\sqrt[2d]{N})$ samples \cite{BhattacharyyaFRV10}.   Acharya et al.~gave an upper bound of $O(\frac{\sqrt{N}}{\epsilon^2} + (\frac{d \log m}{\epsilon^2})^d \cdot \frac{1}{\epsilon^2})$ and a lower bound of $\Omega(\sqrt{N}/\epsilon^2)$  \cite{AcharyaDK15}. While their result gives a tight bound of $\Theta(\sqrt{N}/\epsilon^2)$ when $d$ is relatively small compared to $m$, it does not yield  a tester for Boolean hypercubes using a sublinear number of samples.

Bhattacharyya et al.~considered the problem of monotonicity testing over general posets \cite{BhattacharyyaFRV10}. In particular, they proposed an algorithm for testing the monotonicity of distributions over hypercubes (where $N=2^d$) using $\tilde{O}(N/(\log N/\log \log N)^{1/4})$ samples. They provide a lower bound of $\Omega(n^{1 - o(1)})$ for testing monotonicity of distributions over a matching of size $n$, and a lower bound of $\Omega(\sqrt{n})$ when the poset contains a linear-sized matching in the transitive closure of its Hasse digraph.

In addition to the above, testing monotonicity of distributions has been considered in various settings \cite{AdamaszekCS10, DaskalakisDS12, Canonne15}. There are several works on testing various properties, e.g. uniformity, closeness, and independence  when the underlying distribution is monotone \cite{BatuDKR05, BatuKR04, RubinfeldS05, DaskalakisDSV13,  AcharyaJOS13}.

Testing monotonicity of {\em boolean functions} is also well studied (e.g., \cite{GoldreichGLR98, DodisGLRRS99, LehmanR01, FischerLNRRS02, ChakrabartyS13a, ChakrabartyS14, BelovsB16, BlackCS18}). In the general regime,  the algorithm can query the value of the function at any element in the poset.  This ability is in sharp contrast with our model, in which the algorithm only receive samples according to  the distribution, which do not directly reveal the probability of the elements. 
It is known that one can test monotonicity of functions over hypergrids, and hypercubes using as few  as polylogarithmic queries in the size of the domain. This query complexity is exponentially smaller than the sample complexity of testing monotonicity of distributions, demonstrating that there are  inherent differences between the two problems.

\section{Preliminaries}

We use $[n]$ to indicate the set $\{1, 2, \ldots, n\}$. Throughout this paper we use  the total variation distance denoted by $d_{TV}$ unless otherwise stated. We also denote the $\ell_1$-distance by $d_{\ell_1}$. For a distribution $p$, we denote the probability of the domain element $x$ by $p(x)$. Given a multiset of samples from a distribution on $[n]$, the {\em histogram} of the samples is an $n$-dimensional vector, $h = (h_1, h_2, \ldots, h_n)$, where $h_i$ is the frequency of the $i$-th element in the sample set.

A poset $G = ([n],E)$ is called a {\em line} if and only if $E$ contains all the edges $(i,i+1)$ for $1\leq i \leq n$.
We say a poset is a {\em matching} if all of the edges in the poset are vertex-disjoint. 
We say a poset is {\em bipartite} if the set of vertices can be decomposed in two sets, the {\em top set} and the {\em bottom set}, where no two vertices in the same set are connected. Moreover, 
the direction of {\em all} the edges is from the top set to the bottom set. We use similar terminology for the matching poset as well. In addition, we say a poset $G = (V,E)$ is an {\em $n$-dimensional hypercube} 
when $V$ is $\{0,1\}^n$ and $E$ contains all edges $(u,v)$ 
where there exists a coordinate $i$ such that $u_i=0$ and $v_i=1 $ and $u_j = v_j$ for all $i \neq j$.



\paragraph{Monotonicity.} 
A partially-ordered set (poset) is described as a directed graph $G = (V,E)$, where each edge $(u,v)$ indicates the relationship $u \preceq v$ on the poset. A {\em matching poset} is a poset where the underlying graph $G$ is a matching. A distribution $p$ over a poset domain $V = \{v_1, \ldots, v_n\}$ is a distribution over the vertex set $V$. A distribution $p$ is \emph{monotone} (with respect to a poset $G$) if for every edge $(u, v) \in E$ (i.e., every ordered pair $u \preceq v$), $p(u) \leq p(v)$. Let $\Mon(G)$ be the set of all monotone distributions over the poset $G$. We say that $p$ is \emph{$\epsilon$-far from monotone} if its \emph{distance to monotonicity}, $d_{TV}(p, \Mon(G)) \coloneqq \min_{q \in \Mon(G)} d_{TV}(p,q)$, is at least $\epsilon$. 

\begin{defn}
Let $p$ be a distribution on poset $G$ and $\epsilon$ be the proximity parameter. Suppose an algorithm, $\AA$, has sample access to $p$ and the full description of poset $G$.
$\AA$ is called a {\em monotonicty tester for  distributions} if the following is true with probability at least $2/3$ when the tester has sample access to the distribution.
\begin{itemize}
\item If $p$ is monotone, then $\AA$ outputs \accept.
\item If $p$ is $\epsilon$-far from monotone, then $\AA$ outputs \reject.
\end{itemize}
\end{defn}

\paragraph{Bigness.} A probability distribution $p$ over a domain $[n] = \{1, \ldots, n\}$ is \emph{$T$-big} if, for every domain element $i \in [n]$, $p(i) \geq T$. Related notions for distance to $T$-bigness are defined analogously. The parameter $T$ is called the \emph{bigness threshold}, and may be omitted if it is clear from the context. Let $\Big(n,T)$ indicate the set of all distributions over $[n]$ that are $T$-big. We define the distance to $T$-bigness as $d_{TV}\left(p, \Big(n,T)\right) = \min_{q \in \Big(n,T)} d_{TV}(p,q)$. If this distance is at least $\epsilon$, we say the distribution is $\epsilon$-far from being $T$-big.



\begin{defn}
Let $p$ be a distribution on $[n]$. Suppose Algorithm $\AA$ receives threshold $T$ and bigness parameter $\epsilon$, and has sample access to $p$.  $\AA$ is a {\em $T$-bigness tester} if the following is true with probability at least $2/3$.
\begin{itemize}
\item If $p$ is $T$-big, then $\AA$ outputs \accept.
\item If $p$ is $\epsilon$-far from  $T$-big, then $\AA$ outputs \reject.
\end{itemize}
Also, \emph{$T$-bigness testing problem} refers to the task of distinguishing the above cases with high probability. 
\end{defn}

\begin{remark}
Note that the probability $2/3$ is arbitrary in the above definitions. One can amplify the probability of outputting the correct answer to $1-\delta$ by increasing the number of samples by an  $O(\log 1/\delta)$ factor. 
\end{remark}
\input{overview.tex}

\input{bigness_LB.tex}

\input{Bigness_LB_app.tex}
\input{bigness-to-monotonicity.tex}

\input{reduction_to_bipartite.tex}

\section{Algorithms with Sublinear Sample Complexity} \label{sec:upperbounds}
In this section, we provide sublinear sample complexity algorithms for testing bigness, and testing monotonicity of distributions over different poset domains.
See Section~\ref{subsection:upperbounds} for proof overviews.

\input{bigness_UB.tex}

\input{matchings_UB.tex}

\input{bipartite_UB.tex}

\subsection{Testing monotonicity of distributions that are uniform on a subset of the domain}\label{sec:unif}
In this section, we give an algorithm for testing monotonicity on a specific yet broad class of instances.
More specifically, suppose that we are given a directed bipartite graph $G(V=V_T\cup V_B,E\subseteq V_T\times V_B)$,  along with a probability distribution on the set $V$. Note that all the directed edges go from a vertex in the ``bottom'' set $V_B$, to a vertex in the ``top'' set $V_T$. 
We additionally assume that all distributions which we sample from are \emph{uniform on a subset} of $V$ whose size is known to the algorithm. That is, for every vertex $u\in V$ either $p_u=0$ or $p_u=1/\vert R\vert$, where $R$ is the support of the distribution $p$. 

We will show the following result:
\begin{theorem}\label{thm:UniformOnSubset} Let $G$ be a directed bipartite graph as described above and $p$ be a probability distribution on $V$ which is uniform on a subset of $V$, namely $R$. Given the size of $R$, there exists an algorithm with sample complexity  $O(\frac{n^{2/3}}{\epsilon}+\frac{1}{\epsilon^2})$ 
that can test, with success probability $2/3$, whether $p$ is monotone on $G$, or $p$ is $\epsilon$-far from any monotone function on $G$,
\end{theorem}

At a high level, our tester works as follows: We draw an initial set $\mathcal{S}_1$ of $s_1$ samples from $p$. We define $B=\mathcal{S}_1\cap V_B$ to be the set of vertices from the bottom, $V_B$, that we see in the sample set. Then, we look at the set $T\subseteq V_T$ containing all out-neighbors of the vertices in $B$. 
We show the following structural property of distributions that are $\epsilon$-far from being monotone: in expectation, the constructed set $T$ contains $\epsilon/s_1$ endpoints of violating edges, so  $|T|$ cannot be too small. Thus, if $|T|$ is much smaller than $\epsilon/s_1$, we can immediately conclude that the distribution is close in total variation distance to some monotone distribution.
However, if $T$ is sufficiently large in cardinality, 
we draw more samples in order to estimate the amount of probability mass on $T$. 
Note that if $p$ is monotone, then we expect that all the elements in $T$ be in the support of the distribution, namely $R$, so every single element of $T$ should have probability mass $\frac{1}{\vert R\vert}$ for the distribution to be monotone. The tester rejects if there is sufficient evidence that this is not the case. 
More specifically, the proposed tester is given in Algorithm~\ref{alg:uniform_on_subset}.

\begin{algorithm}[t]\label{alg:uniform_on_subset}
\caption{Algorithm for testing monotonicity of the uniform distribution over a subset of the domain.}
\begin{algorithmic}[1]
\Procedure{Monotonicity-Test}{$G, \epsilon$,$\vert R\vert$, and sample access to $p$}
    \State{$\mathcal{S}_1 \gets$ Draw $s_1=O(\frac{n^{2/3}}{\epsilon})$ samples from $p$.}
    \State {$B \gets  \mathcal{S}_1\cap V_B$ }
    \Comment{where $V_B$ is the set of bottom vertices}
	\State{$T\leftarrow N(B)\quad\quad$}
	\Comment{$N(B)$ is the neighbor set of the set $B$}
    \If{$\vert T\vert\leq \frac{\epsilon s_1}{2}$} 	
        \State {\textbf{Return} {\accept}}
    \EndIf
    \State {$\mathcal{S}_2 \gets$ Draw $s_2=O({n^{2/3}})$ samples from $p$.}
    \State {$Y\gets T\cap \mathcal{S}_2$}
    \State {$\epsilon^\prime\gets \frac{\epsilon\cdot s_1}{2\vert T\vert}$}
    \If{$\vert Y\vert\geq s_2\cdot(1-\frac{\epsilon^\prime}{2})\cdot \frac{\vert T\vert}{\vert R\vert} $} 
		\State{\textbf{Return} {\accept}}
    \Else
    	\State{\textbf{Return} {\reject}}
    \EndIf
\EndProcedure
\end{algorithmic}
\end{algorithm}

\begin{proofof}{Theorem \ref{thm:UniformOnSubset}}
As given in the algorithm, let $s_1=O(\frac{n^{2/3}}{\epsilon})$ and $s_2 = O(n^{2/3})$ denote the sample sizes of the two steps described earlier. We consider the following two cases.

\noindent
\textbf{Completeness case:} Assume $p$ is a monotone distribution. Clearly, each sample we draw has a non-zero probability. Since we pick $T$ to be the neighbor set of the samples we draw, we know that every element in $T$ has a non-zero probability. By the uniformity assumption, this probability is $|T|/|R|$. Thus, when we draw $s_2$ samples from the distribution we expect $|T|/|R|$ fraction of them fall into $T$. So, the expected value of $|Y|$ is $s_2\cdot |T|/|R|$. We defer the asymptotic complexity analysis of this case to the end of our proof. 

\noindent
\textbf{Soundness case:} Assume $p$ is $\epsilon$-far from being a monotone distribution. Consider all the violating edges $(u,v)$ in $E$ for which  $p(u)$ is greater than $p(v)$.  By Lemma~\ref{lem:violating_matching}, there exists a set of edges, namely $M$, that form a matching, and we have: 
$$\sum\limits_{(u, v) \in M} p(u) - p(v) \geq \epsilon \,.$$
Note that without loss of generality one can assume $M$ only has violating edges, since removing non-violating edges only makes the left hand side larger. By the uniformity   assumption for $p$, $p(u) - p(v)$ is exactly $1/|R|$. Thus, by the above inequality, we have $|M|/|R|$ is at least $\epsilon$. 

Since there are $\vert M\vert$ vertices in $V_B$ that belong to the matching,  $\vert B\cap M\vert$ is a random variable distributed according to the binomial distribution $\Bin(s_1,{\vert M\vert}/{\vert R\vert})$, we have that \[\E\left[\vert B\cap M\vert\right]=\frac{s_1 \cdot \vert M\vert}{\vert R\vert} \geq \epsilon s_1 \,.\]
Using Chebyshev's inequality and the fact that $|B\cap M|$ is a binomial distribution, we have 
\begin{align*}
\Pr\left[|B \cap M|  \leq \frac{\epsilon s_1}{2}\right] & \leq \Pr\left[|B \cap M| \leq \frac{\E[\vert B\cap M\vert]}{2}\right] 
\leq \frac{4 \Var[|B\cap M|]}{\E[\vert B\cap M\vert]^2} 
\\ & \leq \frac{4 s_1 \cdot (|M|/|R|) \cdot (1 - |M|/|R|)}{(s_1 \cdot |M|/|R|)^2} \leq \frac{4}{\epsilon s_1} = O(n^{-2/3})\,.
\end{align*}
Thus, with high probability, $B$ contains at least $\epsilon s_1/2$ endpoints in $M$. Note that the neighbor set of $B$ contains the other endpoints of the edges in the matching $M$. Thus, $T$ contains at least $|B\cap M|$ vertices of zero probability, which implies that the size of $T$ has to be at least $\epsilon s_1/2$. 
Hence, for sufficiency large $n$, the probability that $p$ gets rejected due to the condition $|T| \leq \epsilon s_1/2$ is negligible.

Consider the second set of samples we draw in the algorithm $S_2$. Clearly, the size of $Y \coloneqq T \cap S_2$ is a binomial random variable drawn from $\Bin(s_2, |T\cap R|/|R|)$. However, we show
that  $\epsilon \coloneqq \epsilon s_1/(2|T|)$ fraction of the elements in $T$ have zero probability. Thus, $|T \cap R|/|R|$ is at most $(1- \epsilon')|T|/|R|$ while in the completeness case it is $|T|/|R|$. So, we only need to estimate the bias of a Bernoulli random variable up to an additive error of $\epsilon'' \coloneqq \epsilon' |T|/(2|R|)$. By Hoeffding bound, we only need to draw $O(1/{\epsilon''}^2)$ samples to distinguish the two cases with high probability which implies: 
$$s_2 = \Theta\left(\frac{1}{{\epsilon''}^2}\right) = \Theta\left(\frac{|R|^2}{\epsilon'^2 |T|^2}\right) \leq O\left(\frac{n^2}{\epsilon^2 s_1^2}\right) = O (n^{2/3})$$
Thus, with high probability, we distinguish them correctly.  
\end{proofof}

\subsection{Upper bound via trying all matchings}\label{sec:naive}
In this section we present a simple upper bound for the problem of monotonicity testing on bipartite graphs.
Let $\mathcal{M}$ be the number of pairs of subsets $(S_t,S_b)$ of top and bottom elements respectively for which there exists a perfect matching between them.  The algorithm is the following:

\begin{algorithm}[t]\label{alg:bipartite}
\caption{Algorithm for Testing Monotonicity on a bipartite graph.}
\begin{algorithmic}[1]
\Procedure{Matching-Tester}{$G, \epsilon$, and sample access to $p$}
    \State {$s \gets$ draw $O(\log \mathcal{M} / \epsilon^2)$ samples  from $p$.}
    \For{each pair of equal size subsets $(S_t,S_b)$ of top and bottom elements} 
        \If{there exists a perfect matching between $S_t$ and $S_b$} 
            \State {$\hat w_t \gets $ Estimate the total probability mass of $S_t$}
            \State {$\hat w_b \gets $ Estimate the total probability mass of $S_b$}
      	    \If{$\hat w_t$ is less than $\hat w_b-\epsilon/2$}
      			\State{\textbf{Return} {\reject}}
          	\EndIf
        \EndIf 
    \EndFor
    \State{\textbf{Return} {\accept}}
\EndProcedure
\end{algorithmic}
\end{algorithm}

\begin{theorem}
We can test whether a distribution $p$ over a bipartite graph $G$ with $n$ vertices is monotone or $\epsilon$-far from any monotone distribution with success probability $2/3$, using $O((\log M)/\epsilon^2)$ samples, where $M$ is the number of pairs of subsets of top and bottom elements respectively for which there exists a perfect matching between them. That is, $O(n/\epsilon^2)$ samples for a worst case graph $G$.
\end{theorem}
\begin{proof} Let $w_t$ and $w_b$ denote the probability mass of $S_t$ and $S_b$ respectively.
Note that if we use $O(1/\epsilon^2)$ samples, we can estimate $w_t$ and $w_b$  within an additive error of  $\epsilon/8$. Thus, we can estimate the difference of the two with error of $\epsilon/4$ with a constant probability. We can amplify the probability of the correctness, by repeating the estimation and taking the median of them. Therefore, for each pair of subsets, the probability that the algorithm fails to estimate the difference of $w_b$ and $w_t$ within an error of $\epsilon/4$ is at most $O(\frac{1}{M})$. By union bound, we distinguish whether $w_b - w_t$ is at least $\epsilon$ or at most zero by comparing the $\hat w_b - \hat w_t$ with $\epsilon/2$, with a constant success probability.

Now, if $p$ is $\epsilon$-far from being monotone with respect to the graph $G$, there exists a matching such that the total difference between the probabilities of the bottom and the top elements, $w_b - w_t$ is at least $\epsilon$ by Lemma~\ref{lem:violating_matching}. Thus, in one of the iteration, we will consider this matching, and output \reject. 
Also, if $p$ is monotone with respect to the graph $G$, there is no  violating edge. Therefore, for each pair  $S_t$ and $S_b$, we have $w_b -  w_t \leq 0$. Thus, in no iteration we output \reject, and the distribution will be accepted at the end. 

Lastly, since there are at most $2^{n_t}\cdot 2^{n_b}=2^{n_t+n_b}=2^n$ pairs of subsets where $n_t,n_b$ is the total number of top and bottom elements respectively, we conclude that the sample complexity is $O(n/\epsilon^2)$.  
\end{proof}

\paragraph{Remark:} Note that in order to execute the above algorithm, it is not required to know the quantity $M$ in advance.  We can instead draw more samples and update all our estimates at the same time to sufficiently reduce the error probability for each estimate for the union bound to work.

\bibliographystyle{alpha}
\bibliography{main.bib}

\end{document}

%% file: overview.tex
\section{Overview of Our Techniques}
In this section, we give an overview of our results and the high-level idea of our techniques.

\subsection{A lower bound for the bigness testing problem}
\label{subsection:bignesslboverview}

In Section \ref{sec:BignessLB}, we provide two random processes for generating histograms of samples from two families of distributions, such that one family consists of ``big'' distributions, and the other family largely of ``$\epsilon$-far from big'' distributions.
Then, we show that unless a large number of samples have been drawn, the distributions over the histograms generated via these two random processes are statistically very close to each other, and hence appear indistinguishable to any algorithm,
as specified precisely in Theorem~\ref{thm:bigness_LB}.  The construction yields a lower bound for the general problem of testing the bigness property in Corollary~\ref{cor:bigness-test}.  Furthermore, the construction provides  a useful building block for establishing further lower bounds for monotonicity testing in various scenarios in Section~\ref{sec:big2mon}.

To generate histograms from the two families of distributions,  imagine the following process:  We have two prior distributions $P_V$ and $P_{V'}$, and we generate probability vectors (measures), $p$ and $p'$, according to the priors: Each domain element $i$ randomly picks its probability in an i.i.d fashion from the prior distribution. More precisely, let $V_1, V_2, \ldots, V_n$ be $n$ i.i.d. random variables from prior $P_V$, then $p$ is defined to be the following:
$$p = \frac 1 n (V_1, V_2, \ldots, V_n)\,.$$  
We generate $p'$ similarly according to prior $P_{V'}$. While the total probability is unlikely to sum to $1$, we will design the priors, $P_V$ and $P_{V'}$, so that we can later modify $p$ or $p'$ into a probability distribution with only small changes. We then generate histograms of samples from (the normalization of) $p$
by drawing $n$ independent random variables $h_i \sim \Poi(s \cdot p(i))$ (namely $h_i \sim \Poi(sV_i/n)$) for $i = 1, \ldots, n$, and output $h = (h_1, \ldots, h_n)$ as the histogram of the samples. 
Note that by Poissonization method, one may view the histogram as being generated from a set of $\Poi(s \cdot \sum_i {V_i}/n)$ samples from the normalization of $p$. 
Hence, if $\sum_i {V_i}/n$ is close to one, the histogram serves as a set of roughly $s$ samples. 
We set $s$ more specifically in terms of the rest of the parameters later.

The goal in Section~\ref{sec:BignessLB} is to find two prior distributions $P_V$ and $P_{V'}$, then generate two probability vectors $p$ and $p'$, and two histograms $h$ and $h'$ according to them respectively, such that the following events hold with high probability.
\begin{enumerate}
\item 
The probability vectors $p$ and $p'$ are approximate probability distributions; that is, their total probability masses are each close to $1$. 
\item
After scaling the probability vectors $p$ and $p'$ above into respective probability distributions, the normalization of  $p$ is $T$-big, and the normalization of $p'$ is $\epsilon$-far from any $T$-big distribution.
\item
The total numbers of (Poissonized) samples in $h$ and $h'$ drawn from the normalization of $p$ and $p'$ are each $\Omega(s)$, where $s$ is the sample complexity lower bound we are aiming to prove.
\item
Given $h$ or $h'$, distinguishing whether it is generated from $P_V$ or $P_{V'}$ with success probability $2/3$ requires $h$ or $h'$ to contain at least $s$  samples.
\item
Additionally, we will bound the largest probability mass $\pmax$ that the normalized distributions place on any domain element -- this part is not necessary for this section, but will be useful for the reduction between monotonicity testing and bigness testing later on. 
\end{enumerate}

Now, if we choose $P_V$ and $P_{V'}$ carefully such that $h$ and $h'$ are generated according to the above process based on $P_V$ and $P_{V'}$ are hard to distinguish, then we can establish a lower bound for the bigness testing problem. We state this result more formally as the following theorem in Section~\ref{sec:BignessLB}.

\begin{restatable*}{theorem}{bignessLB} \label{thm:bigness_LB}
For 
integer $L = O(\log n)$ and 
sufficiently small $\epsilon=\Omega(1/n)$,
there exist a parameter $\beta = \beta(L,\epsilon)$ 
and two distributions $\HH^+$ and $\HH^-$ over the set of possible histograms of size at least $s = \Omega\left(n^{1-1/L} \log^2 (1/\epsilon)/L\right)$ with the following properties:
\begin{itemize}
\item The histogram generated from $\HH^+$ 
is drawn from a $1/(\beta n)$-big distribution.
\item The histogram generated from $\HH^-$ 
is drawn from a distribution which is $\epsilon$-far from any $1/(\beta n)$-big distribution.
\item $d_{TV}\left(\HH^+, \HH^-\right) \leq 0.01$.
\item The largest probability mass among any elements in any probability distributions above (from which the histograms are drawn) is $\pmax = O(L^2/ (n\log^2(1/\epsilon)))$.
\end{itemize}
\end{restatable*}

An important case of this theorem is when $L = \Theta(\log n)$, where we establish a nearly linear sample complexity lower bound of $\Omega(n/\log n)$ for the general problem of bigness testing as follows.

\begin{restatable*}{corollary}{bignessTest} \label{cor:bigness-test}
For sufficiently small parameter $\epsilon=\Omega(1/n)$, there exists a parameter $\beta = \beta(\epsilon)$ such 
that any algorithm that can distinguish whether a distribution over $[n]$ is $1/(\beta n)$-big or $\epsilon$-far from any $1/(\beta n)$-big distribution with probability $2/3$ requires $\Omega(n \log^2(1/\epsilon)/\log n)$ 
samples. 
In particular when $\epsilon$ is a constant, $\beta$ is constant, then any such algorithm requires $\Omega(n/\log n)$ samples. 
\end{restatable*}

We propose the following optimization problem, $\OP[1]$, such that its optimal solution specifies $P_V$ and $P_{V'}$, satisfying the requirements of the theorem. Intuitively speaking, as $P_V$ aims to generate $T$-big distributions, we must ensure that $V_i$'s are bounded away from $1/\beta$, so that $p(i) = V_i/n$ has expected value higher than $T=1/(\beta n)$. At the same time, we hope to maximize the probability that $V'_i = 0$ so that $p'$ has lots of domain elements with probability zero to make its normalization far from any $T$-big distribution. In addition, we find $P_V$ and $P_{V'}$ under the constraint that the first $L$ moments of them are exactly matched, as to ensure that the resulting distributions over the histograms, $\HH$ and $\HH'$, are statistically close. The objective value of this optimization problem corresponds to the expected distance of $p'$ to the closest $T$-big distribution in the $\ell_1$-distance.

$$
\begin{array}{lll}
{\rm Definition~of~}\OP[1]: &  \sup & \frac 1 \beta \Pr[V' = 0]
\vspace{2mm} \\
& s.t. & \E[V] = \E[V'] = 1
\vspace{2mm} \\
& & \E[V^j] = \E[{V'}^j] \quad \mbox{for } j = 1, 2, \ldots, L
\vspace{2mm} \\
& & V \in \left[\frac {1+\nu} \beta, \frac \lambda \beta \right],
V' \in \{0\} \cup \left[\frac {1+\nu} \beta, \frac \lambda \beta \right]
 \mbox{ and } \beta > 0.
\end{array}
$$
In the above optimization problem, the unknowns are $P_V$, $P_V'$, and $\beta$. $\nu$ and $\lambda$ are two parameters specified latter in the proof. That is we are looking for two distributions $P_V$ and $P_V'$ such that two random variables $V$ and $V'$ drawn from them respectively have expected value one, and their first $L$ moments are matched. Also, $\beta$ controls the range of the probabilities, $p(i)$'s and $p'(i)$'s, and the distance to the bigness property.

 We relate the optimal solution for $\OP[1]$ to an LP defined by \cite{wu2015chebyshev}, who in turn relate their LP to the error from the
best polynomial approximation of the function $1/x$ over the interval $[1+\nu, \lambda]$. By doing this, we show the existence of a solution $(P_V, P_{V'})$ where the value $\Pr[V' = 0]$, which is proportional to the distance to $1/(\beta n)$-bigness in the second family, is sufficiently large. 

Our proof relies on and extends the lower bound techniques for estimating support size provided in \cite{wu2015chebyshev}, incorporating specific conditions for the bigness problem. Firstly, unlike the support size estimation problem, we need our distributions to be fully-supported on the domain $[n]$ for the big distributions, whereas in their case, both families of distributions are allowed to be partially supported. Secondly, our optimization problem treats the threshold $1/(\beta n)$ as a variable, whereas the support size problem simply imposes the strict threshold of $1/n$. Thirdly, based on this construction, we must also give a direct upper bound for the maximum probability, which facilitates our later proofs for providing lower bounds for the matching and hypercube posets.

\subsection{From bigness lower bounds to monotonicity lower bounds}
\label{subsection:monlb}

In Section \ref{sec:big2mon}, we show how to turn our lower bound results for bigness testing problem in Section \ref{sec:BignessLB}, into lower bounds for monotonicity testing in some fundamental posets, namely the matching poset and the Boolean hypercube poset. 

\paragraph{Matching poset.} To establish our lower bound for testing monotonicity of the matching poset, we construct our distribution $p$ by assigning probability masses to the endpoints of edges $(u_i, v_i)$ in our matching as follows: the vertices $u_i$'s are assigned probability masses according to the $T$-bigness construction, whereas the vertices $v_i$'s are uniformly assigned the threshold $T$ as their probability masses; the assigned probabilities are then normalized into a proper probability distribution. We show that before normalization, $p(u_i) \leq T = p(v_i)$ if the original distribution is big; and otherwise, the distance to the monotonicity of the constructed distribution measures exactly the distance to the $T$-bigness property. We then show that the normalization step scales the entire distribution $p$ down by only a constant factor, hence the lower bounds for the monotonicity testing over the matching poset with $2n$ vertices asymptotically preserves the parameters $\epsilon, s$ and $\pmax$ of the lower bound on bigness construction for $n$ domain elements.

\paragraph{Hypercube poset.} To achieve our results for the Boolean hypercube, we \emph{embed} our distributions over the matching poset into two consecutive levels $\ell$ and $\ell-1$ of the hypercube (where $\ell$ denotes the number of ones in the vertices' binary representation). 
We pair up elements in these levels in such a way that distinct edges of the matching have incomparable endpoints: the algorithm must obtain samples of these matched vertices in order to decide whether the given distribution is monotone or not. We also place probability mass $\pmax$ on all other vertices on level $\ell$ and above, and probability mass $0$ on all remaining vertices,
in order to ensure that the distribution is monotone everywhere else.
Lastly, we rescale the entire construction down into a proper probability distribution. Unlike the matching poset, sometimes this scaling factor is super-constant, shrinking the overall distance to monotonicity, $\epsilon$, to sub-constant. Here, we make use of our upper bound on $\pmax$ of the bigness lower bound construction to determine the scaling factor.

\subsection{Reduction from general posets to bipartite graphs}

In Section \ref{sec:reduction_to_bipartite}, we show that the problem of monotonicity testing of distributions over the \emph{bipartite} posets is essentially the ``hardest'' case of monotonicity testing in general poset domains. That is, we show that for any distribution $p$ over some poset domain of size $n$, represented as a directed graph $G$, there exists a distribution $p'$ over a bipartite poset $G'$ of size $2n$ such that (1) $p$ preserves the total variation distance of $p$ to monotonicity up to a small multiplicative constant factor, and (2) each sample for $p'$ can be generated using one sample drawn from $p$. 
These properties together imply the following main theorem of the section.

\begin{restatable*}{theorem}{generaltobipartite}
\label{thm:general2bipartite}
Suppose that there exists an algorithm that tests monotonicity of a distribution over a bipartite poset domain of $n$ elements using $s(n, \epsilon)$ samples for any total variation distance parameter $\epsilon>0$. Then, there exists an algorithm that tests monotonicity of a distribution over any poset domain of $n$ elements using $O(s(2n, \epsilon/4))$ samples.
\end{restatable*}

Our approach may be summarized as follows. We first show, in Theorem~\ref{thm:dist_to_mon_matching}, that we may characterize (up to a constant factor) the distance of $p'$ to monotonicity, as the size of the \emph{maximum matching} on the \emph{transitive closure} of $G$, denoted by $TC(G)$, where the weight $w(u,v) \coloneqq \max\{p(u)-p(v),0\}$ represents the amount that $(u,v)$ is \emph{violating} the monotonicity condition. In particular, we have the following theorem:

\begin{restatable*}{theorem}{distToMonMatching} \label{thm:dist_to_mon_matching}
Consider a poset $G = (E, V)$ and a distribution $p$ over its vertices. Suppose every edge $(u,v)$ in the $TC(G)$ has a weight of $\max(0, p(u) - p(v))$. Then, the total variation distance of $p$  to any monotone distribution is within a factor of two of the weight of the  maximum weighted matching in $TC(G)$.
\end{restatable*}

This crucial theorem provides a \emph{combinatorial} way to approximate the distance to monotonicity for general posets, leading to our upcoming construction of $p'$ for Theorem~\ref{thm:general2bipartite} as well as some algorithms in Section~\ref{sec:upperbounds}. Theorem~\ref{thm:dist_to_mon_matching} is shown via LP duality: the \emph{dual} LP for the problem of optimally ``fixing'' $p$ to make it monotone, turns out to align with the maximum (fractional) matching problem on $G$'s transitive closure. In particular, the dual constraints are of the form $\{Ay \leq b, y \geq 0\}$ where $A$ is a totally unimodular matrix, implying that an \emph{integral} optimal solution exists, namely the maximum matching.

To prove Theorem~\ref{thm:general2bipartite}, given the original poset $G = (V,E)$, we create a bipartite poset with two copies $u^-$ and $u^+$ of each original vertex $u \in V$: the vertices $u^-$'s and $u^+$'s form the bipartition of the new bipartite poset $G'$ of size $2n$. We add $(u^-, v^+)$ to the bipartite poset if $(u,v)$ is in the transitive closure of $G$; that is, there exists a directed path from $u$ to $v$ in $G$. The new probability distribution $p'$ on $G'$, is created from $p$ on $G$, by dividing the probability mass $p(u)$ equally among $p'(u^-)$ and $p'(u^+)$. Note that a sample from $p'$ is obtained by drawing from $p$ and adding the sign $-/+$ equiprobably. It follows via transitivity that $p'$ is monotone over $G'$ when $p$ is monotone over $G$, and via Theorem~\ref{thm:dist_to_mon_matching} that if $p$ is $\epsilon$-far from monotone on $G$, then $p'$ is also at least $\epsilon/4$-far from monotone over $G'$. These conditions allow us to test monotonicity of $p$ on any general poset $G$ by instead testing monotonicity of $p'$ on a bipartite poset $G'$ with parameter $\epsilon' = \epsilon/4$, as desired.

\subsection{Upper bounds results} 
\label{subsection:upperbounds}
In Section~\ref{sec:upperbounds}, we provide sublinear algorithms for testing bigness, and testing monotonicity of distributions over different poset domains.

\paragraph{Bigness testing.} In Section~\ref{sec:bigness-ub}, we provide an algorithm for bigness testing. Observe that the $T$-bigness property is a \emph{symmetric} property: closed under permutation of the labels of the domain elements $[n]$. Hence, we leverage the result of \cite{ValiantV17} that learns the counts of elements for each probability mass: $h_p(x) = |\{a: p(a)=x\}|$. Observe that the distance to $T$-bigness is proportional to the total ``deficits'' of elements with probability mass below $T$. Hence, this learned information suffices for constructing an algorithm for testing bigness, using a sub-linear, $O(\frac{n}{\epsilon^2 \log n})$, number of samples.

\paragraph{Monotonicity testing for matchings.} Next, in Section~\ref{sec:matching-ub}, we provide an algorithm for testing monotonicity of \emph{matching} posets. We again resort to the work of \cite{ValiantV17} for learning the counts of elements for each \emph{pair of probability masses}, with respect to a pair of distributions $p_1, p_2$ over the domain $[n]$, namely $h_{p_1,p_2}(x,y) = |\{a: p_1(a)=x, p_2(a)=y\}|$, given $O(\frac{n}{\epsilon^2 \log n})$ samples each from $p_1$ and $p_2$. We hope to consider our distribution $p$ over a matching $G = (S\cup T, E)$ with $E = \{(u_i,v_i)\}_{i \in [n]} \subset S\times T$ as a pair of distributions, namely $p_S$ and $p_T$, representing probability masses $p$ places over $u_i \in S$ and $v_i \in T$, respectively. Learning $h_{p_S,p_T}$ would intuitively allows us to approximate $p$'s distance to monotonicity by summing up the ``violation'' for pairs $x < y$. However, there are subtle challenges to this approach that do not present in the earlier case of bigness testing.

First, we must somehow rescale $p_S$ and $p_T$ up into distributions according to their total masses $w_S$, $w_T$ placed by $p$. However, it is possible that, say, $p_S = o(1)$, making samples from $S$ costly to generate by drawing i.i.d.~samples from $p$. We resolve this issue via a reduction to a different distribution $p'$ that approximately preserves the distance to bigness, while placing comparable total probability masses to $S$ and $T$. Second, the algorithm of \cite{ValiantV17} learns $h_{p_1,p_2}(x,y)$ according to a certain distance function, that we must lower-bound by the total variation distance. In particular, this bound must be established under the presence of errors in the scaling factor, as $w_S$ and $w_T$ are not known to the algorithm. We overcome these technical issues, which yields an algorithm for testing monotonicity over matchings. We maintain the same asymptotic complexity as that of \cite{ValiantV17}. 

\paragraph{Monotonicity testing for bounded-degree bipartite graphs.} Moving on, in Section~\ref{sec:bipartite-ub}, we tackle the problem of monotonicity testing in \emph{bipartite} posets; as shown in Section~\ref{sec:reduction_to_bipartite}, this bipartite problem captures the monotonicity testing problem of \emph{any} poset. We make progress towards resolving this problem by offering our solution for the \emph{bounded-degree} case. We turn the distribution $p$ on a bipartite poset $G$ of maximum degree $\Delta$, into a distribution $p'$ on a \emph{matching} poset $G'$ that approximately preserves the distance to monotonicity: applying the algorithm of Section~\ref{sec:matching-ub} above constitutes a monotonicity test for $p$ with sample complexity $O(\frac{\Delta^3 n}{\epsilon^2 \log n})$.

Our reduction simply places $\Delta$ copies $v_1, \ldots, v_\Delta$ of each vertex $v \in V(G)$ into $V(G')$, then for each edge $(u,v) \in E(G)$, connects a pair of unused endpoints $(u_i,v_j$), as to create a matching subgraph of size $|E(G)|$ on $G'$. The probability distribution $p'$ on $V(G')$ simply distributes probability mass $p(v)$ equally among all $\Delta$ copies $v_i$'s. (Each remaining, isolated vertex is matched with a dummy $0$-mass vertex, turning $G'$ into a matching poset.) This new graph $G'$ contains $O(\Delta n)$ vertices, and we show that $d_{TV}(p',\Mon(G')) \geq d_{TV}(p, \Mon(G)) / (2\Delta)$ by explicitly creating a ``low-cost'' scheme for ``fixing'' $p$ into a monotone distribution on $G$, based on the optimal scheme that turns $p'$ monotone on $G'$, charging at most an extra $2\Delta$-multiplicative factor.

\paragraph{Testing monotonicity of distributions that are uniform on a subset of the domain. } In Section \ref{sec:unif}, we show that for a specific broad family of distributions on directed bipartite graphs of arbitrary degree, we can test monotonicity of such distribution using $O(\frac{n^{2/3}}{\varepsilon}+\frac{1}{\varepsilon^2})$ samples. Namely, our result applies for 
distributions that are uniform on an arbitrary subset of the domain, 
given that every poset edge is directed from some vertex in the ``bottom'' part to some vertex in the ``top'' part of the graph. 
Our tester performs roughly the following: First, we sample a number of vertices from the graph and throw away ones that lie in the top part. For the remaining ones in the bottom part, denoted $B$, we identify their neighbors $T$ in the top part, and determine whether or not they all belong to the support of the distribution. Since the distribution is uniform in its support, this condition is sufficient for the distribution to be monotone in the induced subgraph $G[B\cup T]$. The tester accepts when it cannot rule out the possibility that $T$ has the maximum possible probability mass. Recall that if the distribution is $\epsilon$-far from monotone, there must exist a large matching of ``violated'' edges. To this end, we show that the induced subgraph $G[B\cup T]$ contains many disjoint violated edges, implying that there are many vertices in $T$ outside of the support: the probability mass on $T$ will be noticeably small and the tester will reject.

\paragraph{Upper bound via trying all matchings.} In Section \ref{sec:naive} we give another upper bound for testing monotonicity of a distribution with respect to a bipartite graph which, in this case, has a small number of induced subgraphs that contains a perfect matching of their vertices. In particular, we show that $O(\frac{\log M}{\epsilon^2})$ samples are sufficient for this task, where $M$ is the number of such induced subgraphs.  We note that this bound matches the general learning upper bound of $O(n/\epsilon^2)$ when $M$ attains its maximum value of $2^{\Theta(n)}$, but can potentially be better when $M$ is asymptotically smaller.    
The main idea of our tester is as follows: if the distribution is $\epsilon$-far from monotone, there exists a matching of violated edges that is $\Theta(\epsilon)$-far from monotone. Hence, for each subgraph of $G$ that admits a perfect matching, we may approximate the weight (violation amount) of this matching by simply comparing the total probability masses between the top part and the bottom part of the subgraph. We approximate these masses with error probability $O(1/M)$ for each subgraph, which allows us to apply a union bound over all subgraphs at the end. Our tester rejects if the weight of one such subgraph exceeds $\epsilon$, or accepts otherwise.

%% file: bigness_LB.tex
\section{A Lower Bound for the Bigness Testing Problem} \label{sec:BignessLB}

In this section, we give a lower bound for the bigness testing problems. As described in the overview in Section ~\ref{subsection:bignesslboverview}, 
 we provide two random processes for generating samples from two families of distributions, such that one family consists of ``big'' distributions, and the other family largely of ``$\epsilon$-far from big'' distributions, and then show
 that they are hard to distinguish.

First, we define a random process that, given a prior distribution, $P_V$,  over non-negative numbers, generates a random probability distribution over the domain elements $[n]$, and then draws samples from it. More specifically, let $V$ be a random variable drawn from $P_V$, and we also use $P_V$ to denote the probability density function (PDF) over $V$; for now we require $\E[V]=1$, and will specify further desired properties momentarily. We generate an \emph{approximate} probability distribution $p$ according to $P_V$. The distribution $p$ is constructed by having each domain element $i$ choose its probability $p(i)$, in an i.i.d.~fashion, from the prior distribution, $P_V$,  over possible probabilities. Then, we construct a histogram of roughly $s$ samples from $p$ according to the following steps:
\begin{itemize}
\item \textbf{Step 1:}
Generate $n$ i.i.d.~random variables $V_1, V_2, \ldots, V_n$ according to $P_V$, then form the following \emph{probability vector} over $[n]$:
$$p = \frac 1 n (V_1, V_2, \ldots, V_n)\,.$$
Remark that, while $p$ is not necessarily a probability distribution under this notion, the condition $\E[V]=1$ suggests that the total probability masses of $p$ is likely to be centered around $1$. So, $p$ is likely to be approximately a probability distribution, and can be normalized into one while modifying individual entries $p(i)$'s by only a small multiplicative factor.

\item \textbf{Step 2:} Draw $n$ independent random variables $h_i \sim \Poi(s \cdot p(i))$ (namely $h_i \sim \Poi(sV_i/n)$) for $i = 1, 2, \ldots, n$, and output $h = (h_1, h_2, \ldots, h_n)$ as the histogram of the samples. While we do not explicitly normalize $p$, since $p$ is an approximate probability distribution, this histogram still captures (with high probability) $\Omega(s)$ Poissonized samples drawn from the normalization of $p$.
\end{itemize}

The goal in this section is to find two prior distributions $P_V$ and $P_{V'}$, to generate two probability vectors $p$ and $p'$ according to the above process such that after the normalization, $p$ and $p'$ have the desire properties: $p$ is  big (every $p(i)$ is at least the threshold $T$), and $p'$ is $\epsilon$-far from any big distribution ($p'$ contains a significant number of entries $i$ with $p'(i)=0$). 
Then, we generate two histograms $h$ and $h'$ according to $p$ and $p'$ respectively. 
If the histograms $h$ and $h'$ are hard to distinguish, then we can establish a lower bound for the bigness property. This requirement will show up as constraints for designing two prior distributions, $P_V$ and $P_{V'}$, to achieve these families of distributions with high probability. 
Below, we summarize the conditions that we need the prior distributions to hold (with high probability):

\begin{enumerate}
\item 
The probability vectors $p$ and $p'$ are approximate probability distributions;
that is, all of their coordinates are non-negative and their total probability masses are each close to one.

\item
After scaling the probability vectors $p$ and $p'$ above into respective probability distributions, the normalization of $p$ is $T$-big, and the normalization of $p'$ is $\epsilon$-far from any $T$-big distribution.
\item
The total numbers of (Poissonized) samples in $h$ and $h'$ drawn from the normalization of $p$ and $p'$ are each $\Omega(s)$.
\item
Given $h$ or $h'$, distinguishing whether it is generated from $P_V$ or $P_{V'}$ with success probability $2/3$ requires $h$ or $h'$ to contain a large number of samples.
\item
Additionally, we will bound the largest probability mass $\pmax$ that the normalized distributions place on any domain element -- this part is not necessary for this section, but will be useful for the reduction between monotonicity testing and bigness testing later on. 
\end{enumerate}
We state this result as the following theorem.

\bignessLB

\begin{proof}
Let positive values $\nu$, $\lambda$, $\beta$, and a positive integer $L$ be a set of parameters with the following property 
that we determine more precisely later: 
$$
0 < \nu \leq \frac{1}{2}, \quad \lambda > 1 + \nu, \quad 1\leq\beta\leq\min\left\{\frac{1}{\epsilon}, \lambda \right\} \; \mbox{ and } \; L = O(\log n).
$$
Throughout this section, we consider the bigness threshold $T=1/(\beta n)$, and note that the value $\beta$ itself may depend on the error parameter $\epsilon$, an the number of matched moments $L$. 
Note also that $\nu$ is a constant. 

We propose the following optimization problem, $\OP[1]$, such that its optimal solution, specifying $P_V$ and $P_{V'}$ satisfies the requirements of the theorem. 
Recall that $p$ and $p'$ are generated by drawing $n$ i.i.d samples, $V_i$'s and $V'_i$'s, from $P_V$ and $P_{V'}$ respectively:
\begin{align*}
p = \frac 1 n (V_1, V_2, \ldots, V_n) \quad \quad \quad p' = \frac 1 n (V'_1, V'_2, \ldots, V'_n)
\end{align*}
Intuitively speaking, as $P_V$ aims to generate $T$-big distributions, we must ensure that the $V_i$'s are  bounded away from $1/\beta$, so that $p(i) \sim V_i/n$ has expected value higher than $T=1/(\beta n)$. At the same time, we hope to maximize the probability that $V'_i = 0$ so that $p'$ is far from any $T$-big distribution, under the constraint that the first $L$ moments of $P_V$ and $P_{V'}$ are exactly matched, as to ensure that the resulting distributions of histograms $\HH$ and $\HH'$ are statistically close. The objective value of this optimization problem corresponds to the expected distance of $p'$ to the closest $T$-big distribution in total variation distance. To clarify the notation, $\lambda$ and $\nu$ are given to us. The unknown variables in $\OP[1]$ are the PDFs $P_V$ and $P_{V'}$ of two random variables $V$ and $V'$, respectively, as well as the scaling variable $\beta>0$. The parameter $\lambda$ roughly specifies the ratio between the largest and the smallest non-zero probabilities that $p$ and $p'$ can take.\footnote{
Note that $P_V$ and $P_{V'}$ are on a continuous domain. However, $P_{V'}$ will additionally have a non-negligible probability \emph{mass} placed at value $0$. In fact, it turns out that in the optimal solution, $P_V$ and $P_{V'}$ are only supported on a few distinct values ($\Theta(L) = O(\log n)$ of them), so the optimal $P_V$ and $P_{V'}$ assume the role of probability mass functions rather than PDFs.}

\begin{equation} \label{eq:op1}
\begin{array}{lll}
{\rm Definition~of~}\OP[1]: &  \sup & \frac 1 \beta \Pr[V' = 0]
\vspace{2mm} \\
& s.t. & \E[V] = \E[V'] = 1
\vspace{2mm} \\
& & \E[V^j] = \E[{V'}^j] \quad \mbox{for } j = 1, 2, \ldots, L
\vspace{2mm} \\
& & V \in \left[\frac {1+\nu} \beta, \frac \lambda \beta \right],
V' \in \{0\} \cup \left[\frac {1+\nu} \beta, \frac \lambda \beta \right]
 \mbox{ and } \beta > 0.
\end{array}
\end{equation}

In the following lemma, we find the optimal value of $\OP[1]$. We use $\OPT(A)$ to refer to the optimal value of optimization problem $A$. 
\begin{restatable}{lemma}{optLP}
\label{lem:optLP1}
 For any $\nu$ and $\lambda$ such that $0 < 1+\nu < \lambda$, there exists a scaling  parameter, $\beta$, in $[1+\nu, \min(\lambda, 1/\OPT(\OP[1]))]$ such that 
$$
\OPT(\OP[1]) = \left(\frac {1}{\sqrt{1 + \nu}}- \frac 1 {\sqrt \lambda}\right)^2 \left(\frac{\sqrt{\frac \lambda {1 + \nu}} - 1}{\sqrt{\frac \lambda {1 + \nu}} + 1}\right)^{L-2} \,.
$$
\end{restatable}
The proof of Lemma~\ref{lem:optLP1} is postponed to Section \ref{sec:optLP1}.

Let the value of  $\beta$ be determined by the above lemma, and set $d$ to be $\OPT(\OP[1])$.

Recall our wish list of five properties for the priors, $P_V$ and $P_V'$, that we propose in the introduction  of Section~\ref{sec:BignessLB}. We define the following ``good" events , which hold with high probability, to formalize the properties of the generated vectors $p$ and $p'$. 

$$E = \left\{
\left|\sum\limits_{i=1}^n \frac{V_i}n - 1\right| \leq \nu, \mbox{ and }
\sum\limits_{i=1}^n N_i > s(1-\nu)/2
\right\} \,.
$$
and
$$
E' \
= \left\{
\left|\sum\limits_{i=1}^n \frac{V'_i}n - 1\right| \leq \nu\, , 
\, r \geq \frac {\beta n d}{2}, \mbox{ and }
\sum\limits_{i=1}^n N'_i > s(1-\nu)/2
\right\}
$$
where $r$ is the number of elements $i$ such that $V'_i$ is zero. Roughly speaking, these events state that $p = \frac{1}{n}(V_i)_{i\in[n]}$ and $p'=\frac{1}{n}(V'_i)_{i\in[n]}$, generated in step 1, are approximate probability distributions (having total masses in $[1-\nu,1+\nu] = \Theta(1)$), and step 2 generates sufficient numbers of samples in the histogram (at least $s(1-\nu)/2 = \Omega(s)$ each). Further, $p'$ consists of as many as $r \geq \beta n d / 2$ elements with probability mass $0$, thus is at distance at least $rT \geq d/2$ from any $T$-big distribution -- we will set $d \geq 2\epsilon$ to reach the desired result.

In the following lemma, we show that conditioning on $E$ and $E'$, after running the process using the priors $P_V$ and $P_{V'}$, the generated histogram $h$ is a sufficiently large set of samples from a $1/(\beta n)$-big distribution, and histogram $h'$ is a sufficiently large set of samples from a distribution which is $\epsilon$-far from any $1/(\beta n)$-big distribution. In addition, the total variation distance between the distribution over $h$'s and $h'$'s is bounded when $P_V$, $P_{V'}$ form a solution of $\OP[1]$. More precisely, let $\HH$ denote the distribution over histograms $h$ generated by the process when the prior is $P_V$, and let $\HH_E$ be the distribution over histograms $h$ conditioning on $E$. We define $\HH'$ and $\HH_{E'}$ similarly. In the following lemma, we bound the total variation distance between $\HH_E$ and $\HH'_{E'}$ as well.

\begin{restatable}{lemma}{dtvBound}
\label{lem:dtvBound}
Let $P_V, P_{V'}$, and $\beta \in [1, 1/d]$
form a solution of $\OP[1]$ with objective value $d \geq 2 \epsilon$. Suppose $P_V$ and  $P_{V'}$ are the prior distributions to generate histograms $h$ and $h'$ according to the process. Then, $h$ given event $E$ is a histogram of a set of at least $s(1-\nu)/2$ samples from a $1/(\beta n)$-big distribution, whereas $h'$ given $E'$ is a histogram of a set of at least  $s(1-\nu)/2$ samples that are drawn from a distribution which is $\epsilon$-far from any $1/(\beta n)$-big distribution. Moreover, 
$$d_{TV}\left(\HH_{E}, \HH'_{E'}\right) \leq  \frac{2\lambda}{\beta n \nu^2} +\exp\left( - \frac {\beta n d}{8}\right) + 2 \exp{\left(-\frac{s(1-\nu)}{6}\right)} + n \left(\frac{e s\lambda}{2 n L}\right)^L
\, .$$
Lastly, the largest probability mass among any elements in any probability distributions (from which the samples are drawn) is $\lambda/(n(1-\nu))$.
\end{restatable}
The proof of Lemma~\ref{lem:dtvBound} is given in Section \ref{sec:dtvBound_proof}.

Now, we assign the parameters, $\nu, \lambda$, and $s$, as follows:
\begin{align*}
\nu \coloneqq 1/2, \quad  
\lambda \coloneqq (1+\nu)\cdot \left(\frac{4(L-2)}{\ln\left( 1/(27\epsilon)\right)}-1\right)^2\, , \mbox{ and }
 s \coloneqq  \left\lfloor \frac {L \, n} {2 e \lambda}\right\rfloor \,
\end{align*}
Recall that we set $d$ to be the optimal value of $\OP[1]$, and Lemma~\ref{lem:optLP1} tells us its value. We show that in this setting $d$ is at least $2\epsilon$. Let $\rho$ be $\sqrt{\lambda/(1+\nu)}$. Then, we have:
\begin{align*}
d & \coloneqq  \left(\frac {1}{\sqrt{1 + \nu}}- \frac 1 {\sqrt \lambda}\right)^2 \left(\frac{\sqrt{\frac \lambda {1 + \nu}} - 1}{\sqrt{\frac \lambda {1 + \nu}} + 1}\right)^{L-2} 
\geq \frac{1}{1+\nu}\left(1 - \frac{1}{\sqrt{\frac{\lambda}{1 +\nu}}}\right)^2 \left(\frac{\sqrt{\frac{\lambda}{1 +\nu}} - 1}{\sqrt{\frac{\lambda}{1 +\nu}} + 1}\right)^{L-2} 
\\
& = \frac 2 3 \left(1 - \frac 1 \rho\right)^2 \left(1-\frac 2 {\rho+1}\right)^{L-2} > \frac 2 {27}  \left(\frac 1 {e^2}\right)^{\frac{2(L-2)}{\rho+1}} \geq \frac 2 {27} \exp 
\left(-\frac{4(L-2)}{\rho+1}\right)\geq 2 \epsilon
\,.
\end{align*}
as long as $\rho \geq 1.5$. 
It is not hard to see that, for sufficiently large $n$ and $\epsilon \geq c/n$ for sufficiently large constant $c$, then $\rho \geq 1.5$ holds, yielding $d \geq 2\epsilon$,
{
for every $\epsilon \leq c_0$, where $c_0 < 1/2$ is a constant.}

Let $\HH^+$ and $\HH^-$ be $\HH_{E}$ and $\HH'_{E'}$ respectively. By Lemma~\ref{lem:dtvBound}, the total variation distance between $\NN^+$ and $\NN^-$ is at most $0.01$, while $s$ and $\pmax$ behave according to the claimed respective asymptotic bounds. Hence, the proof is complete. 
\end{proof}

An important case of Theorem~\ref{thm:bigness_LB} is when $L = \Theta(\log n)$, where we establish a near-linear sample complexity lower bound of $\Omega(n/\log n)$ for the general problem of bigness testing as follows.


\bignessTest

\begin{proof}
By Theorem \ref{thm:bigness_LB}, there exist $\HH^+$ and $\HH^-$ with the aforementioned properties. Any $1/(\beta n)$-bigness tester has to distinguish between $\HH^+$ and $\HH^-$ with probability at least 2/3. On the other hand, the total variation distance between $\HH^+$ and $\HH^-$ is at most 0.01. Therefore, no algorithm can distinguish between them while receiving  $s/4 = \Theta(n\log^2(1/\epsilon)/\log n)$ samples with probability more than $(1 + 0.01)/2$. Therefore, testing $1/(\beta n)$-bigness requires $\Omega(n\log^2(1/\epsilon)/\log n)$ samples. 

Note that in the proof of Theorem \ref{thm:bigness_LB}, $\beta$ is determined by Lemma \ref{lem:optLP1}, and it is bounded by $1/\epsilon$. Thus, if $\epsilon$ is a constant then $\beta$ is also a constant. Thus, the required sample complexity becomes $\Omega(n/\log n)$.
\end{proof}

%% file: Bigness_LB_app.tex
\subsection{Proof of Lemma \ref{lem:optLP1}}
\label{sec:optLP1}
\optLP*

\begin{proof}
To prove the lemma, we introduce an auxiliary linear program (\LP[2]) that is known to have an optimal value of the right hand side of the above equation. We prove the $\LP[2]$ has the same optimal objective value as $\OP[1]$ to prove the lemma. For two given parameters $\nu$ and $\lambda$, we define  the following LP over two random variables $X, X'$. 
\begin{equation}\label{eq:LP2}
\begin{array}{lll}
{\rm Definition~of~}\LP[2]: &  \sup & \E\left[\frac 1 X\right] - \E\left[\frac 1 {X'}\right]
\vspace{2mm} \\
& s.t. & \E[X^j] = \E[{X'}^j] \quad \mbox{for } j = 1, 2, \ldots, L - 1
\vspace{2mm} \\
& & X, X' \in \left[1+\nu, \lambda \right]
\end{array}
\end{equation}
To interpret this LP, assume the unknown variable is the $PDF$'s of the random variables $X$ and $X'$. Thus, for any number $x$ in $[1+ \nu, \lambda]$, we want to find $P_X(x)$ and $P_{X'}(x)$. Note that this optimization problem is linear since all the expectations above are a linear function of $P_X$ and $P_{X'}$. Moreover, there is an implicit constraint here that the integral of $P_X$ and $P_{X'}$ should be one since they are probability distributions.

Observe that there exists a trivial solution where $X$ and $X'$ are two identically-distributed random variables, so $\LP[2]$ is feasible and its optimal objective value is at least zero. Let $\XX^*$ and $\XX'^*$ be a pair of random variables forming an optimal solution for $\LP[2]$, and let $\beta^* = 1/\E[1/\XX^*]$. 
Since all $X$ and $X'$ are in $[1 +\nu, \lambda]$, then $\beta^*$ is also in $[1+\nu, \lambda]$. On the other hand, since $\XX'^*$ is positive and bounded, then $\E[1/\XX'^*] > 0$ and thus $\E[1/\XX^*] > \OPT(\LP[2])$; hence $\beta^*$ is at most $1/\OPT(\LP[2])$.

Now, we argue that $\LP[2]$ and $\OP[1]$ have the same optimal value. We introduce two new random variables $\VV^*$ and $\VV'^*$ with the following PDFs, and later we show they form an optimal solution for $\OP[1]$. 
\begin{align*}
P_{\VV^*} (v) &\coloneqq \frac {{\beta^*}} v \, P_{\XX^*}({\beta^*} v) + \left(1 - {\beta^*}\, \E \left[\frac 1 {\XX^*}\right]\right) \delta_0(v), \; \mbox{ and } \;\\
P_{\VV'^*} (v) &\coloneqq \frac {\beta^*} v \, P_{{\XX'}^*}({\beta^*} v) + \left(1 - {\beta^*} \, \E \left[\frac 1 {{\XX'}^*}\right]\right) \delta_0(v)
\end{align*}

In the above equations, with a slight abuse of notation we say that $1/v$ is zero for $v=0$; that is, the probability \emph{mass} for $v=0$ is given by the respective second terms. Since ${\beta^*}$ is defined to be $1/\E[1/\XX^*]$, the second term in $P_{\VV^*}$ is zero for all $v$ in particular for $v=0$. We define our notation in this fashion in order to make the calculations for $\VV^*$ and $\VV'^*$ analogous, so we may write our proof compactly.

Now, we show that the proposed variables $\VV^*$, $\VV'^*$ and ${\beta^*}$ form a feasible solution for $\OP[1]$. First, we show that the domain of $\VV^*$ and $\VV'^*$ are as stated in the definition of $\OP[1]$ in Equation~\ref{eq:op1}. Then, we show $P_{\VV^*}$ and $P_{\VV'^*}$ are probability distribution, and we prove the constraints of $\OP[1]$ hold as well.

First, consider the domain of the random variables. Clearly the domain does not include the numbers where the PDF is zero, so we prove that the $P_{\VV^*}$ and $P_{{\VV'}^*}$ are (potentially) non-zero only when when $\VV^*$ and $\VV'^*$ are in the rage specified by the domain constraints of the $\OP[1]$. 
Recall that the second term in $P_{\VV^*}$ is always zero. Thus, $P_{\VV^*}$ could be potentially non-zero only if $x$ equal to $\beta v$ has a non-zero error probability according to $P_{\XX^*}$. Therefore, $\VV^*$ is always in $[(1+\nu)/{\beta^*}, \lambda/{\beta^*}]$. For $\VV'^*$, in addition to the value $v\in[(1+\nu)/{\beta^*}, \lambda/{\beta^*}]$, $v$ could be zero as well since the second term in the definition of $P_{\VV'^*}$ may be non-zero at $v=0$. Thus, $\VV'^*$ is always in $\{0\}\cup [(1+\nu)/{\beta^*}, \lambda/{\beta^*}]$.

In addition, $P_{\VV^*}$ (and similarly $P_{{\VV'}^*}$) is a probability distribution since the integral of the PDF is one:
\begin{align*}
\int_{-\infty}^{\infty} P_{\VV^*}(v) dv & = \int_{(1+\nu)/{\beta^*}}^{\lambda/{\beta^*}}\frac {\beta^*} v \,  P_{\XX^*}({\beta^*} v) dv + \left(1 - {\beta^*}\, \E \left[\frac 1 {\XX^*}\right]\right) \int_{-\infty}^{\infty}\delta_0(v) dv
\\ & = 
\int_{1+\nu}^{\lambda}  \frac {{\beta^*}^2} x \,  P_{\XX^*}(x) \cdot \frac 1 {\beta^*} \, dx + \left(1 - {\beta^*}\, \E \left[\frac 1 {\XX^*}\right]\right)
\\ & = 
{\beta^*}\, \E \left[\frac 1 {\XX^*} \right]+ \left(1 - {\beta^*}\, \E \left[\frac 1 {\XX^*}\right]\right) = 1,
\end{align*}
where the second equality is derived by substituting $v$ with $x/{\beta^*}$.

Now, we focus on the constraints of $\OP[1]$. The first constraint is $\E[\VV^*] = \E[\VV'^*] = 1$. Below we show that the expected value of $\VV^*$ is $1$. 
\begin{align*}
\E[\VV^*] & = 
\int_{-\infty}^{\infty} v \, P_{\VV^*}(v) dv  = \int_{(1+\nu)/{\beta^*}}^{\lambda/{\beta^*}}{\beta^*} \,  P_{\XX^*}({\beta^*} v) dv + \left(1 - {\beta^*}\, \E \left[\frac 1 {\XX^*}\right]\right) \int_{-\infty}^{\infty}v \, \delta_0(v) dv
\\ & = 
\int_{1+\nu}^{\lambda}  {\beta^*} \,  P_{\XX^*}(x) \cdot \frac 1 {\beta^*} \, dx = 1
\end{align*}
One can similarly show that $\E[{\VV'}^*]=1$, and the constraint holds. 

The second constraint is that the first $L$ moments of $\VV^*$ and $\VV'^*$ are matched: $\E[{\VV^*}^j] = \E[{\VV'^*}^j]$ for $j$ in $[L]$. The previous constraint implies that the first moments, $\E[{\VV}^*]$ and $\E[{\VV'}^*]$, are equal, so here we focus on the second and higher moments. Fix $j$ in $\{2, \ldots, L\}$. For the $j$-th moment of $\VV^*$, we have:
\begin{align*}
\E[{\VV^*}^j] & = 
\int_{-\infty}^{\infty} v^j \, P_{\VV^*}(v) dv  = \int_{(1+\nu)/{\beta^*}}^{\lambda/{\beta^*}}{\beta^*}\, v^{j-1} \,  P_{\XX^*}({\beta^*} v) dv + \left(1 - {\beta^*}\, \E \left[\frac 1 {\XX^*}\right]\right) \int_{-\infty}^{\infty}v^j \, \delta_0(v) dv
\\ & = 
\int_{1+\nu}^{\lambda} \frac{x^{j-1}}{{\beta^*}^{j-2}} \,  P_{\XX^*}(x) \cdot \frac 1 {\beta^*} \, dx = \frac{1}{{\beta^*}^{j-1}} \, \E[{\XX^*}^{j-1}].
\end{align*}
We can similarly show the same condition for $\E[{\VV'^*}^j]$. Since $\XX^*$ and $\XX'^*$ satisfies the moment matching constraints of $\LP[2]$, we derive the moment matching constraints of $\OP[1]$ as follows:
$$
\E[{\VV'^*}^j] = \frac{1}{{\beta^*}^{j-1}} \, \E[{\XX^*}^{j-1}] = \frac{1}{{\beta^*}^{j-1}} \,\E[{\XX'^*}^{j-1}] = \E[{\VV'^*}^j] \, .
$$
Therefore, $\VV^*$, $\VV'^*$ and ${\beta^*}$ form a feasible solution for $\OP[1]$. Thus, the objective function according to $\VV^*$, $\VV'^*$ is at most the optimal value of $\OP[1]$:
$$
\OPT(\OP[1]) \geq \frac 1 {\beta^*} \, \Pr[\VV'^* = 0]
$$
On the other hand, the objective value of $\OP[1]$ and $\LP[2]$ are the same on the two solutions we discussed:
$$
\frac 1 {\beta^*} \, \Pr[\VV'^* = 0] = \frac 1 {\beta^*} \left(1 - {\beta^*} \, \E \left[\frac 1 {\XX'^*}\right]\right) = \E \left[\frac 1 {\XX^*}\right] - \E \left[\frac 1 {\XX'^*}\right] = \OPT(\LP[2])
$$
where the last equality is true, since we chose $X$ and $X'$ to be the optimal solution of $\LP[2]$ at the beginning.

\begin{equation}\label{eq:vv*=opt2}
\OPT(\OP[1]) \geq \frac 1 {\beta^*} \, \Pr[\VV'^* = 0] = \OPT(\LP[2])
\, .
\end{equation}


We continue the proof by showing that the above inequality is true in the other direction, i.e, $\OPT(\OP[1])$ is at most $\OPT(\LP[2])$. Let $P_\VV$, $P_{\VV'}$ and $\beta$ form a feasible solution for $\OP[1]$. We define random variables $\XX$ and $\XX'$ with the following PDFs, and show that they form a feasible solution for $\LP[2]$ in Equation~\ref{eq:LP2} with the same objective value as $\VV$ and $\VV'$ in the $\OP[1]$:
$$
P_\XX(x) \coloneqq \frac x {\beta^2} P_{\VV} \left(\frac{x}{\beta}\right), \mbox{ and }\ \  \
P_{\XX'}(x) \coloneqq \frac x {\beta^2} P_{\VV'}  \left(\frac{x}{\beta}\right).
$$
First, we show that the domain of $\XX$ and $\XX'$ matches with the domain constraint in $\LP[2]$. Similar to the previous part, we prove that the PDF's are zero outside the interval specified by the domain constraint $[1+\nu, \lambda]$.  Observe that $P_\XX(x)$ is non-zero if and only if $x$ and $P_\VV(x/\beta)$ are both non-zero, so $x/\beta$ has to be in $[(1+\nu)/\beta, \lambda/\beta]$. Thus, the domain of the random variable $\XX$ (and similarly $\XX'$) is $[1+\nu, \lambda]$. 

Moreover, note that $P_\XX$ (and similarly $P_{\XX'}$) is a probability distribution:
$$
\int_{-\infty}^{+\infty} P_{\XX}(x) dx = \int_{1+\nu}^{\lambda} \frac x {\beta^2} \cdot P_{\VV}\left(\frac x \beta\right) dx = \int_{(1+\nu)/\beta}^{\lambda/\beta} \frac v \beta \cdot P_{\VV}(v) \cdot \beta dv = \E[\VV] = 1 \, 
$$
where the equation is derived by replacing $x/\beta$ with a new variable $v$. Now, we show that the constraints of $\LP[2]$ are satisfied for $\XX$ and $\XX'$. Fix $j \in [L-1]$. We show the $j$-th moment of $\XX$ and $\XX'$ are equal:

\begin{align*}
\E[\XX^j] = \int_{-\infty}^{+\infty} x^j P_{\XX}(x) dx = \int_{1+\nu}^{\lambda} \frac {x^{j+1}}{\beta^2}  \cdot P_{\VV}\left(\frac x \beta\right) dx = \int_{(1+\nu)/\beta}^{\lambda/\beta} \frac {\beta^j \, v^{j+1}} {\beta} \cdot P_{\VV}(v) \cdot \beta dv = \beta^j \, \E[\VV^{j+1}] \, .
\end{align*}
Similarly, one can show $\E[{\XX'}^j]$ is equal to $\beta^j \E[{\VV'}^{j+1}]$. Since the pair $\VV$ and $\VV'$ satisfies the moment matching constraints of $\OP[1]$, then $\E[\VV^{j+1}]$ is equal to $\E[{\VV'}^{j+1}]$. Therefore, $\E[\XX^j]$ is equal to $\E[{\XX'}^j]$.

Now, we focus on the objective functions of the $\OP[1]$ and $\LP[2]$. We have:
\begin{align*}
\E\left[\frac 1 \XX\right] - \E\left[\frac 1 {\XX'} \right] & = 
\int_{1+\nu}^{\lambda} \frac 1 x \cdot P_{\XX}(x) dx - \int_{1+\nu}^{\lambda} \frac 1 {x'} \cdot P_{\XX'}(x') dx'
\\ & = 
\int_{1+\nu}^{\lambda} \frac 1 {\beta^2} \cdot P_{\VV}\left(\frac x \beta\right) dx 
- \int_{1+\nu}^{\lambda} \frac 1 {\beta^2} \cdot P_{\VV'}\left(\frac {x'} \beta\right) dx' 
\\ & =
\int_{(1+\nu)/\beta}^{\lambda/\beta} \frac {1} {\beta^2} \cdot P_{\VV}(v) \cdot \beta dv 
- \int_{(1+\nu)/\beta}^{\lambda/\beta} \frac {1} {\beta^2} \cdot P_{\VV'}(v') \cdot \beta dv' 
\\ & = \frac {1} {\beta} \cdot\left( \int_{(1+\nu)/\beta}^{\lambda/\beta}  P_{\VV}(v) \cdot dv 
- \int_{(1+\nu)/\beta}^{\lambda/\beta}  P_{\VV'}(v') \cdot dv' \right)
\\ & = \frac 1 \beta \cdot \left(1 - \left(1 - \Pr[\VV' = 0]\right)\right) = \frac 1 \beta \Pr[\VV' = 0] \,.
\end{align*}
Now that for any feasible solution of $\OP[1]$, there exists a feasible solution for $\LP[2]$ with the same objective value, one can conclude that the optimal value of $\OP[1]$, $\OPT(\OP[1])$, is at most $\OPT(\LP[2])$. Thus by Equation \ref{eq:vv*=opt2}, we have:
$$\OPT(\OP[1]) \geq \frac 1 {\beta^*} \, \Pr[\VV'^* = 0] = \OPT(\LP[2]) \geq \OPT(\OP[1])$$
which implies that $\VV^*$, $\VV'^*$ and $\beta^*$ also form an \emph{optimal} solution for $\OP[1]$, and hence $\OPT(\OP[1])$ and $\OPT(\LP[2])$ are equal. This also implies that $\beta^*$ is at most $1/\OPT(\OP[1])$.

In Appendix~E of \cite{WuY16Entropy}, Wu and Yang proved that an optimal solution of $\LP[2]$ can be obtained through the best polynomial approximation of the function $1/x$. More formally, they showed that there exists a solution for $\LP[2]$ with the following optimal value: 
$$
\OPT(\LP[2]) = 2 \inf\limits_{f \in \PP_{L-1}} \sup\limits_{x \in [1+\nu, \lambda]} \left|\frac 1 x - f(x)\right|
$$
where $\PP_{L-1}$ is the set of all degree $L-1$ polynomials. The optimal polynomial approximation error have been studied 
in \cite{KrausVZ12} and in Sec. 2.11.1 of \cite{timan63}. They computed the maximum error of the best degree $L-1$ polynomial approximation. More precisely, we have:
$$
\OPT(\OP[1]) = \OPT(\LP[2]) = 2 \inf\limits_{f \in \PP_{L-1}} \sup\limits_{x \in [1+\nu, \lambda]} \left|\frac 1 x - f(x)\right| = 
\left(\frac {1}{\sqrt{1 + \nu}}- \frac 1 {\sqrt \lambda}\right)^2 \left(\frac{\sqrt{\frac \lambda {1 + \nu}} - 1}{\sqrt{\frac \lambda {1 + \nu}} + 1}\right)^{L-2} \, .
$$
Hence, the proof is complete.
\end{proof}

\subsection {Proof of Lemma \ref{lem:dtvBound}}
Before stating the lemma, we review the definitions we used so far. Recall that $p$ and $p'$ are generated by drawing $n$ i.i.d samples, $V_i$'s and $V'_i$'s, from $P_V$ and $P_{V'}$ respectively:
\begin{align*}
p = \frac 1 n (V_1, V_2, \ldots, V_n) \quad \quad \quad p' = \frac 1 n (V'_1, V'_2, \ldots, V'_n)
\end{align*}
and $E$ and $E'$ where the desired events:
$$E = \left\{
\left|\sum\limits_{i=1}^n \frac{V_i}n - 1\right| \leq \nu, \mbox{ and }
\sum\limits_{i=1}^n N_i > s(1-\nu)/2
\right\} \,.
$$
and
$$
E'
= \left\{
\left|\sum\limits_{i=1}^n \frac{V'_i}n - 1\right| \leq \nu\, , 
\, r \geq \frac {\beta n d}{2}, \mbox{ and }
\sum\limits_{i=1}^n N'_i > s(1-\nu)/2
\right\}
$$
where $r$ was the number of elements $i$ for which $V'_i$ is zero. 
We generate histograms $h$ and $h'$ according to $p$ and $p'$ respectively. let $\HH$ denote the distribution over histograms $h$ generated by the process when the prior is $P_V$, and let $\HH_E$ be the distribution over histograms $h$ conditioning on $E$. We define $\HH'$ and $\HH_{E'}$ similarly. In the following lemma, we prove ``good properties'' for $p$ and $p'$ after normalization and also bound the total variation distance between $\HH_E$ and $\HH'_{E'}$.

\label{sec:dtvBound_proof}
\dtvBound*

\begin{proof}
First, we show given event $E$, the normalization of $p$ is $1/(\beta n)$-big distribution. From $\OP[1]$, we know that the $V_i$'s are in $[(1+\nu)/\beta, \lambda/\beta]$, and the $V'_i$'s are in $ \{0\} \cup [(1+\nu)/\beta, \lambda/\beta]$. Observe that $p(i)$ after normalization is at least the following: 
$$
\frac {p(i)}{\sum_j p(j)} \geq \frac{V_i/n}{\sum_j V_j/n} \geq \frac{(1+\nu)/(\beta n)}{\sum_j V_j/n} \geq \frac{1}{\beta n}
$$
where the last inequality is due to the fact that $\sum_j V_j/n$ is at most $1+\nu$. Thus, the normalization of $p$ is $1/(\beta n)$-big. On the other hand, we can achieve the same lower bound for the normalized value of $p'(i)$ when $V'_i$ is not zero, so the normalization of $p'$ places either probability mass zero, or at least $1/(\beta n)$, on each element.{
Similarly, the maximum probability mass among the normalization of $p$'s and $p'$'s is at most
$$\frac {p(i)}{\sum_j p(j)} \leq \frac{V_i/n}{\sum_j V_j/n} \leq \frac {\lambda/(\beta n)}{1-\nu} \leq \frac{\lambda}{n(1-\nu)}$$
because $\beta \geq 1$, yielding the desired bound on the maximum probability mass.}

Next, we show that
given $E'$, the normalization $p'$ is $\epsilon$-far from any big distribution. Note that if $V'_i$ is zero, then probability $p'(i)$ even after normalization remains zero. So, there are exactly $r$ elements that have probability mass zero and the rest (based on above argument) each have probability mass at least $1/(\beta n)$. Thus, the total variation distance to $1/(\beta n)$-bigness is at least $r /(\beta n)$, and given $E'$ it is at least 
$d/2\geq \epsilon$.

Now, we show the distance between $\HH_E$ and $\HH'_{E'}$ is bounded. By the triangle inequality we have: 
\begin{align*}
d_{TV}\left(\HH_E, \HH'_{E'}\right) & \leq d_{TV}\left(\HH_E, \HH\right) + d_{TV}\left(\HH, \HH'\right) + d_{TV}\left(\HH', \HH'_{E'}\right) 
\\ & \leq \Pr[E^c] + d_{TV}\left(\HH, \HH'\right) + \Pr[E'^c] 
\, ,
\end{align*}
where the superscript $c$ for the events, $E$ and $E'$ indicates the complimentary event. Now, we start with bounding  the probability of the complementary events of $E$ and $E'$ from above to show that they happen with small probability. Since the $V_i$'s (and similarly the $V'_i$'s) are independently drawn from $P_V$ with expected value $1$, and they are in the range $[0,\lambda/\beta]$, then by the Chebyshev inequality, we have:
$$\Pr\left[\left|\sum\limits_{i=1}^n \frac{V_i}{n} - 1\right| > \nu \right] \leq \frac{\sum_i \Var[V_i]}{n^2\nu^2}\leq \frac{\E[V^2]}{n \nu^2} \leq  \frac{\lambda\E[V]}{\beta n \nu^2} \leq  \frac{\lambda}{\beta n \nu^2}.$$

Recall that the $d$ was the optimal value of $\OP[1]$. Thus, $\Pr[V'_i = 0]$ is $\beta d$. Moreover, $r$, the number of the $V'_i$'s that are zero, is a Binomial random variable with $\E[r] = n\cdot\Pr[V'_i = 0]$ which is $\beta n d$. Thus, by the Chernoff bound, we have:
\begin{align*}
\Pr\left[r < \frac{\beta n d}{2} \right] & = \Pr \left[ \frac {r}{n} < \beta d \left(1 - \frac 1 2\right)\right] \leq \exp\left( - \frac {\beta n d}{8}\right).
\end{align*}

Finally, we show that the total number of samples is high with high probability. Assume we already have $\sum_{i=1}^n V_i/n$ is at least $1-\nu$. Then the total number of samples $\sum_{i=1}^n h_i$ is a Poisson random variable with mean $t \coloneqq s\sum_{i=1}^n V_i \geq s(1-\nu)$. By the tail bound for Poisson distributions proved in \cite{Cannone17}\footnote{If $X$ is a Poisson random variable with mean $\lambda$, then for any $t>0$, we have $\Pr\left[X \leq  \lambda - t\right] \leq  \exp\left({-\frac{t^2}{\lambda+t}}\right)$}, 
we have 
\begin{align*}
\Pr\left[\sum\limits_{i=1}^n h_i  \leq \frac{s(1-\nu)}{2}\right] & \leq 
\Pr\left[\sum\limits_{i=1}^n h_i \leq t - t/2 \right] 
 \leq \exp\left(-\frac{(t/2)^2}{t + t/2}\right) \\ & \leq 
\exp{\left(-\frac{t}{6}\right)} \leq \exp{\left(-\frac{s(1-\nu)}{6}\right)}
\end{align*}
One can achieve a similar result for $\sum_{i=1}^n h'_i$.

Now, we continue bounding the distance between $\HH_E$ and $\HH'_{E'}$. $\HH^{(i)}$ (and similarly ${\HH'}^{(i)}$) indicates the distribution over the $i$-th coordinate of the histogram, $h_i$. By the previous inequality, we have:
\begin{align*}
d_{TV}\left(\HH_E, \HH'_{E'}\right)  & \leq \Pr[E^c] + d_{TV}\left(\HH, \HH'\right) + \Pr[E'^c] 
\\ & \leq \Pr[E^c] + \Pr[E'^c] + n \cdot d_{TV}\left(\HH^{(i)} ,{\HH'}^{(i)}\right) 
\\ & \leq  \frac{2\lambda}{\beta n \nu^2} +\exp\left( - \frac {\beta n \epsilon}{8}\right) + 2 \exp{\left(\frac{s(1-\nu)}{6}\right)} + n \left(\frac{e s\lambda}{2 n L}\right)^L
\, ,
\end{align*}
where the last inequality follows from the fact that the first $L$ moments of $P_V$ and $P_{V'}$ are matched, by Lemma 6 in \cite{wu2015chebyshev}, 
we have:
$$d_{TV}\left(\NN_{V}^{(i)} , \NN_{V'}^{(i)} \right) \leq \left(\frac{e s\lambda}{2 n L}\right)^L$$ 
Hence, the proof is complete. 
\end{proof}

%% file: bigness-to-monotonicity.tex
\section{From Bigness to Monotonicity} \label{sec:big2mon}
In this section, we show how to turn our lower bound results for bigness testing problem in the previous section, into lower bounds for monotonicity testing in some fundamental posets, namely the matching poset and the Boolean hypercube poset. 
See Section \ref{subsection:monlb} for the proof overviews.

\subsection{Monotonicity testing on a matching poset}

\begin{theorem}\label{thm:big2mon}
Consider the pair of distributions $\NN^+$, $\NN^-$ for the bigness problem as specified in Theorem~\ref{thm:bigness_LB} with bigness threshold $T = O(1/n)$, number of samples $s$, and maximum probability $\pmax$. There exists a distribution on a matching of size $n$ with maximum probability $\pmax'=\Theta(\pmax)$ such that testing, with success probability $2/3$, whether a matching randomly drawn from such a distribution is monotone or $\epsilon'=\Theta(\epsilon)$-far from any monotone distribution, requires $s'=\Omega(s)$ samples.
\end{theorem}

{\color{blue}
}

\begin{proof}
Let $U = \{u_1, u_2, \ldots, u_n\}$ and $V = \{v_1, v_2, \ldots, v_n\}$ form the vertex set of a directed matching $M_n$ of size $n$ where the edges are $(v_i, u_i)$'s for $i=1, 2, \ldots, n$. Consider the distribution over the matching poset $G = (U \cup V, \{(v_i, u_i)| i \in [n]\})$; more specifically, the distribution is monotone if and only if the probabilities $p(u_i) \geq p(v_i)$ for all $i$. We apply the Poissonization technique, then prove our lower bound by contradiction: assume there exist an algorithm $\mathcal{A}$ which tests monotonicity of distributions over the matching of size $n$ using $\Poi(s')$ samples where $s' = o(s)$ and successfully distinguishes whether the distribution is monotone or $\epsilon'\coloneqq\epsilon/(2(1+nT))$-far from monotone with probability at least 2/3. To reach the desired contradiction, we turn these samples into $s'(1+nT)$ samples for the $T$-bigness testing problem, and show that one can test $T$-bigness using $\mathcal{A}$ as a black-box tester. Note that $T = O(1/n)$, so the factor $1+nT$ is $\Theta(1)$ in this proof.

Assume we have a distribution, $p$, over $[n]$ elements for which we wish to test the bigness property. We construct a distribution $q_p$ over a matching over $U \cup V$ based on $p$ as follows:
\[q_p(u_i) = \frac{p(i)}{1+nT}, \quad q_p(v_i) = \frac{T}{1+nT}.\]
Clearly the maximum probability of $q_p$ is at most $\pmax'\coloneqq\pmax/(1+nT)$. Next we show the changes in distances to monotonicity. 
Next we show the difference in distance to monotonicity from the case that $p$ is $T$-big and the case that $p$ is $\epsilon$-far from $T$-big. If $p$ is a $T$-big distribution, then $q_P(u_i) \geq T/(1+nT) \geq q_p(v_i)$ and thus $q_p$ is monotone. 

Next, if $p$ is $\epsilon$-far from any $T$-big distribution, then we show that $q_p$ is $\epsilon/(2(1+nT))$-far from any monotone distribution. Let $S$ be the set of elements for which $p(i) < T$. Clearly, to make $p$ a $T$-big distribution, one has to increase all the $p(i)$ to $T$ for $i \in S$ and there is no need to increase the probability of any other elements. Therefore, the total variation distance to of $p$ to $\Big(n)$ is exactly $\sum_{i\in S} T-p(i)$ assuming $T\leq 1/n$. Let $q'$ be the closest monotone distribution  to $q_p$, and observe that $q'(u_i) \geq q'(v_i)$. We compute:
\begin{align*}
d_{TV} (q_p, \Mon(M_n)) & = d_{TV}(q_p, q') = \frac 1 2 \sum\limits_{i=1}^n |q_p(u_i) - q'(u_i)| + |q_p(v_i) -  q'(v_i)|
\\
&=  \frac 1 2 \sum\limits_{i=1}^n  \left|\frac{p(i)}{1+nT} - q'(u_i)\right| + \left|\frac{T}{1+nT} - q'(v_i)\right|  
\\
&\geq \frac 1 2  \sum\limits_{i = 1}^n \max\left(0, q'(u_i) - \frac {p(i)}{1+nT} + \frac {T}{1+nT} - q'(v_i)\right) 
\\&\geq \frac 1 2  \sum\limits_{i = 1}^n \max\left(0, \frac {T-p(i)}{1+nT}\right) 
 \geq \frac{1}{2(1+nT)} \sum\limits_{i \in S} T -p(i)
\\ & = \frac{d_{TV}(p,\Big(n))}{2(1+nT)}= \frac{\epsilon}{2(1+nT)}\, .
\end{align*}


Finally we show that the assumed algorithm $\mathcal{A}$ may be used to test the $T$-bigness property of $p$. Suppose we are given access to $\Poi(s')$ independent samples from the distribution $p$ for which we want to test $T$-bigness property. We construct a distribution $q_p$ as described above: to obtain $\Poi(s'(1+nT))$ samples from $q_p$, for each $i \in [n]$, we create $\Poi(s'\cdot p(i))$ and $\Poi(s'\cdot T)$ samples of $u_i$ and $v_i$ respectively. The $\Poi(s'\cdot p(i))$ samples for each $i$ of the $u_i$'s may be obtained by substituting each element $i$ from $p$ with $u_i$ in $\Poi(s')$ samples from $p$, whereas $\Poi(s'\cdot T)$ samples for $v_i$'s may be generated directly by drawing $v_i$'s uniformly at random. Thus, using $\Poi(s')$ samples from $p$, one can construct $\Poi(\Omega(s'))$ samples from $q_p$ and use $\mathcal{A}$ for testing the monotonicity of the matching poset $q_p$, which corresponds to testing the $T$-bigness of $p$, yielding a contradiction by the fact that bigness testing requires $\Omega(s)$ samples by Theorem~\ref{thm:bigness_LB}.
\end{proof}

This result, applied with Theorem~\ref{thm:bigness_LB} using $L = \Theta(\log n)$ (where $s = \Omega\left(n \ln^2 (1/\epsilon)/\log n\right)$, $\pmax = O((\log^2 n)/ (n\ln^2(1/\epsilon)))$ and $T = 1/(\beta n) \in [\epsilon/n, 1/n]$), immediately yields the following lower bound for the testing monotonicity in a matching poset.

\begin{corollary} \label{cor:matching-lb}
For sufficiently small parameter $\epsilon=\Omega(1/n)$, any algorithm that can distinguish whether a distribution over a matching poset on $2n$ vertices is monotone, or $\epsilon$-far from any monotone distribution, with probability $2/3$ requires $\Omega((n \ln^2 (1/\epsilon))/\log n)$ samples. Moreover, the maximum probability mass of the distribution in the lower bound construction can be bounded above by $O((\log^2 n)/ (n\ln^2(1/\epsilon)))$.
\end{corollary}

\subsection{Monotonicity testing on a hypercube poset}

Consider the Boolean hypercube poset $\{0, 1\}^d$ with $N=2^d$ vertices. For convenience, let $\mathcal{C}$ and $\mathcal{S}$ denote the distribution of distributions implicitly constructed in the lower bound of Theorem~\ref{thm:big2mon}, where distributions in $\mathcal{C}$ are monotone, and distributions in $\mathcal{S}$ are $\epsilon$-far from any monotone distribution, respectively. Theorem~\ref{thm:big2mon} shows that randomly-drawn distributions from $\mathcal{C}$ and $\mathcal{S}$ generate statistically similar histograms over the matching poset. For simplicity, we do not distinguish the parameters $\epsilon$, $s$ and $\pmax$ in Theorem~\ref{thm:bigness_LB} and Theorem~\ref{thm:big2mon} as they are equivalent up to a constant factor.

\subsubsection{General lower bound for monotonicity testing on a hypercube poset}

We first establish the theorem that describes the result of the outlined embedding approach, then later apply this result to achieve interesting special cases.

\begin{theorem}\label{thm:mon2hyp}
Let an integer $\ell \geq 1$ be a parameter. Suppose that there exists a pair $(\mathcal{C},\mathcal{S})$ of distribution of distributions over a matching on $n={\binom{d-1}{\ell-1}}$ pairs of vertices, forming an instance for the monotonicity problem with distance $\epsilon$, a maximum probability $\pmax$, and a lower bound of $s$ samples. Then, testing monotonicity on the Boolean hypercube of size $N = 2^d$ with distance parameter $\epsilon / W$ requires $\Omega(sW)$ samples, where $s = \Omega((n \ln^2 (1/\epsilon))/\log n)$ and $W = 1 + \Theta((\log^2 n)/ (n\ln^2(1/\epsilon))) \cdot \left(\sum_{i=\ell}^{d} {\binom{d}{i}} - {\binom{d-1}{\ell-1}}\right)$.

\end{theorem}

\begin{proof}
Consider two consecutive levels $\ell$ and $\ell-1$ of a hypercube, where the $\ell^\textrm{th}$ level consists of vertices whose coordinates contain exactly $\ell$ ones. 
Our approach is to embed our matching onto these levels in the hypercube, so that each edge of the matching has one endpoint in each of the two levels, and each endpoint is mutually incomparable to any endpoint of any other edge. 


We choose our coordinates for the embedding as follows. We pick all the vertices such that there are exactly $\ell-1$ ones among the first $d-1$ coordinates. Let $M$ denote the set of these vertices. 
There are exactly $2\cdot{{d-1} \choose {\ell-1}}$ vertices in the set $M$. 
Clearly, each vertex in $M$ is comparable with the vertex whose coordinate only differs at the last bit. Furthermore, it is incomparable with the rest of the vertices in $M$, as other coordinates also have $\ell-1$ ones on the first $d-1$ bits.

Next we describe the probabilities assigned to each vertex on the hypercube, given $p$, the distribution over a matching (drawn from $\mathcal{C}$ or $\mathcal{S}$). First we assign the probabilities to $M$ according to $p$. Namely, the set of coordinates of $M$ with $\ell$ ones corresponds to $U$ and that with $\ell-1$ ones corresponds $V$, where $U$ and $V$ are as defined in the previous proof. 
Then, for the remaining vertices in level $\ell$ and above, assign the probability of $c \cdot ((\log^2 n)/ (n\ln^2(1/\epsilon)))$ for a sufficiently large $c$ such that the quantity becomes at least $\pmax$.
Let $W = 1 +  \Theta((\log^2 n)/ (n\ln^2(1/\epsilon))) \cdot \left(\sum_{i=\ell}^{d} {d \choose i} - {{d-1} \choose {\ell-1}}\right)$ be the total probability assigned to all these vertices so far. We divide all assigned probabilities by $W$ to finally obtain a distribution over the hypercube. We denote the constructed distribution over the hypercube $p_{H}$. 

Clearly, the proposed construction preserves the monotonicity due to the incomparability between distinct embedded matching edges. In particular, if distribution over the matching is drawn from $\mathcal{C}$, the distribution over the hypercube will still be monotone; if it is drawn from $\mathcal{S}$, then the distance to monotonicity is now $\epsilon/W$
since, at the very least, the subposet restricted to the embedded matching must be modified to a monotone distribution over this matching.

Using Corollary~\ref{cor:matching-lb}, any algorithm that can test the monotonicity of $p_{H}$ requires $\Omega(s)$ samples from the matching vertices. Note that if we draw a sample from $p_{H}$ with probability $1/W$ it is from the matching. Therefore, observe that $\Poi(s)$ samples from the matching are required in order to obtain $\Poi(sW)$ samples from the hypercube with high probability. This yields the lower bound of $\Omega(sW)$ samples for testing monotonicity over the hypercube poset.
\end{proof}

\subsubsection{Applications of Theorem~\ref{thm:mon2hyp}}

We extend Theorem~\ref{thm:mon2hyp} into two following corollaries. Firstly, we consider embedding our matching to the largest possible levels of the hypercube, namely the middle ones, showing the lower bound of $\Omega(nd)$ samples for $\epsilon = \Theta(1/d^{2.5})$ (Corollary~\ref{cor:hyp1}). To complement this first corollary that only handles sub-constant $\epsilon$, we secondly apply our construction to higher levels of the hypercube, and readjust the construction from Theorem~\ref{thm:bigness_LB} so that $L = \Theta(1)$ moments are matched (as opposed to $\Theta(\log n)$). This approach shows the lower bound of $\Omega(N^{1-\delta})$ for testing monotonicity on the hypercube poset with distance parameter $\epsilon$, such that $\delta \rightarrow 0$ as $\epsilon \rightarrow 0$ (Corollary~\ref{cor:hyp2}).

\begin{corollary} \label{cor:hyp1}
For sufficiently small $\epsilon = \Theta(1/d^{2.5})$, any algorithm that can distinguish whether a distribution over a Boolean hypercube poset of size $N = 2^d$ is monotone, or $\epsilon$-far from any monotone distribution, with success probability $2/3$ requires $\Omega(N d)$ samples.
\end{corollary}

\begin{proof} \;
Let $\ell$ be $\ceil{d/2}$.
As we stated in the proof of Theorem~\ref{thm:mon2hyp}, we embed a matching of size $n \coloneqq {d-1\choose \ell-1}$ onto the middle layer of the hypercube where $n$ is at least $\Omega(N/\sqrt{d}) = \Omega(N/\sqrt{\log N})$ by Stirling's approximation. We have 
$$W = 1+\Theta(d^{2.5}/(N \log^2(1/\epsilon')))  \cdot \Theta(N) = \Theta(d^{2.5}/\log^2(1/\epsilon'))\,.$$ Applying Theorem~\ref{thm:mon2hyp}, we achieve our lower bound of $\Omega(Nd)$ for $\epsilon = \Theta(1/d^{2.5})$ by choosing a sufficiently small constant $\epsilon'$.
\end{proof}

\begin{corollary} \label{cor:hyp2}
Any algorithm that can distinguish whether a distribution over a Boolean hypercube poset of size $N = 2^d$ is monotone, or $\epsilon$-far from any monotone distribution, with success probability $2/3$ requires $\Omega(N^{1-\delta})$ samples, where $\epsilon$ and $\delta=\Theta(\sqrt{\epsilon})+o(1)$ are constants. In particular, $\delta \rightarrow 0$ as $\epsilon \rightarrow 0$.
\end{corollary}

\begin{proof} 
Without loss of generality assume $d$ is even. Otherwise, observe that when $d$ is odd, we may embed a hypercube of size $2^{d}$ in a hypercube of size $2^{d+1}$ and achieve the same lower bound up to a constant factor. 
 Consider $\ell \geq d/2$. Observe that
\[{d \choose {\ell+i}} = {d \choose \ell}\cdot\frac{d-\ell}{\ell+1}\cdot\frac{d-\ell-1}{\ell+2}\cdots\frac{d-\ell-i+1}{\ell+i} \leq {d \choose \ell}\left(\frac{d-\ell}{\ell+1}\right)^i.\]
This yields the inequality
\[\sum_{i=\ell}^{d} {d \choose i} = {d \choose \ell} \sum_{i=0}^{d-\ell} \left(\frac{d-\ell}{\ell+1}\right)^i \leq {d \choose \ell} \sum_{i=0}^{\infty} \left(\frac{d-\ell}{\ell+1}\right)^i = {d \choose \ell} \frac{\ell+1}{2\ell-d+1}.\]
We pick $\ell = d/2 + \alpha d$ for some constant $0.24 >\alpha>0$ so that $\sum_{i=\ell}^{d} {d \choose i} = \Theta\left({d \choose \ell}\right)$. The embedded matching is of size $n = {{d-1} \choose {\ell-1}} = \frac{d}{\ell}{d \choose \ell} = \Theta\left({d \choose \ell}\right)$.

Next, consider the application of Theorem~\ref{thm:big2mon} leveraging Theorem~\ref{thm:bigness_LB} with constant parameters $\epsilon$ and $L$, yielding the lower bound of $s=\Omega(n^{1-1/L}/L)$ samples for $\pmax = O(L^2/n) = O(L^2/{d \choose \ell})$. We compute $W = 1+\Theta\left(L^2/{d \choose \ell}\right)\cdot\Theta\left({d \choose \ell}\right) = \Theta(L^2)$. Applying Theorem~\ref{thm:mon2hyp}, we achieve the lower bound of $\Omega(n^{1-\frac{1}{L}}L)$ for testing monotonicity over the hypercube with $\epsilon = \Theta(1/L^2)$.

Recall that $\ell = d/2 + \alpha d$. Using a similar argument as above, we can also bound
\begin{align*}
n &= {d \choose \ell} \geq {d \choose {d/2}} \cdot \frac{d/2 - 1}{d/2 + 1} \cdots \frac{d/2 - \alpha d}{d/2 + \alpha d} \geq {d \choose {d/2}}\left(\frac{d/2 - \alpha d}{d/2 + \alpha d}\right)^{\alpha d} \\ &\geq {d \choose {d/2}}(1-4\alpha)^{\alpha d} \geq \frac{N}{\sqrt{2d}}(1-4\alpha)^{\alpha \log N} = \frac{N^{1+\alpha \log (1-4\alpha)}}{\sqrt{2d}},
\end{align*}
establishing the lower bound of $\widetilde{\Omega}(N^{(1+\alpha \log (1-4\alpha))(1-\frac{1}{L})}) = \Omega(N^{1-\delta})$ for testing monotonicity over the hypercube poset, where $\delta = 1/L-\alpha \log (1-4\alpha)+o(1)$. Since $\epsilon = \Theta(1/L^2)$, for sufficiently large $N$, we may choose sufficiently small $\alpha$ and large $L$, so that $\delta = \Theta(\sqrt{\epsilon})+o(1)$, as desired.
\end{proof}

%% file: reduction_to_bipartite.tex
\section{Reduction from General Posets to Bipartite Graphs} \label{sec:reduction_to_bipartite}

In this section, we show that the problem of monotonicity testing of distributions over the \emph{bipartite} posets is essentially the ``hardest'' case of monotonicity testing in general poset domains. That is, we show that for any distribution $p$ over some poset domain of size $n$, represented as a directed graph $G$, there exists a distribution $p'$ over a bipartite poset $G'$ of size $2n$ such that (1) $p$ preserves the total variation distance 
of $p$ to monotonicity up to a small multiplicative constant factor, and (2) each sample for $p'$ can be generated using one sample drawn from $p$. 
These properties together imply the following main theorem of this section.


\generaltobipartite

\begin{proof}
Consider an arbitrary poset described as a directed graph $G = (V, E)$, and an associated probability distribution $p$ over $V$. We construct a bipartite graph $G' = (V', E')$ based on the transitive closure of $G$, denoted by $TC(G)$, and a distribution $p'$ over $V'$ such that testing  the monotonicity of $p$ over $V$ is roughly equivalent to testing the monotonicity of $p'$ over $V$. 

The construction of the bipartite $G'=(V',E')$ is as follows: for each $v \in V$, we add two vertices $v^+$ and $v^-$ to $V'$, so that $S \coloneqq \{v^+\}_{v \in V}$ and $T \coloneqq \{v^-\}_{v\in V}$ together form the bipartition $V' \coloneqq S \cup T$. Think of $S$ and $T$ as the set of top and bottom vertices respectively. Next, consider two vertices $u$ and $v$ such that there is a path from $u$ to $v$ in $G$ (i.e., $(u, v)$ is an edge in $TC(G)$).
For every such pair, we add the directed edge $(u^-,v^+)$ to $E'$. Given the distribution $p$ over $V$, we set $p'(v^+) = p'(v^-) = p(v)/2$. Observe that we can generate a sample from $p'$ using a sample from $p$: if $v$ is drawn from $p$, a sample for $p'$ is obtained by picking either $v^+$ or $v^-$, each with probability $1/2$. 

Now, we prove that testing monotonicity of $p$ is equivalent to testing monotonicity of $p'$. If $p$ is monotone, then $p'$ is also monotone: for each $(u^-, v^+)\in E'$, $p(u) \leq p(v)$ via the transitivity of monotonicity of $p$ along the $u$-$v$ path on $G$. So, $p'(u^-) = p(u)/2 \leq p(v)/2 = p'(v^+)$.

Next, suppose $p$ is $\epsilon$-far from $p'$. By Lemma \ref{lem:violating_matching} (shown below), there exists a (directed) matching $M$ in $TC(G)$, such that 
\begin{equation} \label{eq:violating_matching}
\sum_{(u, v) \in M} p(u) - p(v) \geq d_{TV}(\Mon(G), p) \geq  \epsilon \, .
\end{equation}
Then, the set of edges $(u^-, v^+)$'s corresponding to $(u, v) \in M$ also forms a matching, $M'$,  on $G'$. Let $p'^*$ be the monotone distribution on $G'$ closest to $p'$. Since $p'^*$ is a monotone distribution, for an edge $(u^-,v^+)$, $p'^*(v^+)$ is at least $p'^*(u^-)$. Then, by the triangle inequality, we obtain:

\begin{equation*} \label{eq:distToMatching} 
\begin{array}{ll}
d_{TV}(\Mon(G'), p') & = \frac 1 2 \cdot |p' - p'^*| =  \frac 1 2\sum\limits_{v \in V} |p'(v^-) - p'^*(v^-)| + |p'(v^+) - p'^*(v^+)| 
\\
& \geq  \frac 1 2 \sum\limits_{(u^-, v^+) \in M'} |p'(u^-) - p'^*(u^-)| + |p'(v^+) - p'^*(v^+)|
\\
& \geq  \frac 1 2 \sum\limits_{(u^-, v^+) \in M'} p'(u^-) - p'^*(u^-) - p'(v^+) + p'^*(v^+)
\\
& =   \frac 1 2 \sum\limits_{(u^-, v^+) \in M'} p'(u^-) - p'(v^+) + \left(p'^*(v^+) - p'^*(u^-)  \right)
\\
& \geq   \frac 1 2 \sum\limits_{(u^-, v^+) \in M'} p'(u^-) - p'(v^+) 
\\
& \geq  \frac 1 2 \sum\limits_{(u, v) \in M} (p(u) - p(v))/2 \geq \epsilon/4.
\end{array}
\end{equation*}
Note that the second to last inequality is true since $p'^*$ is monotone, and $p'^*(v^+)$ has to be at least $p'^*(u^-)$. Therefore, if $p$ is $\epsilon$-far from monotone, then $p'$ is $\epsilon/4$-far from monotone. 

Thus, to distinguish whether $p$ is monotone or $\epsilon$-far from any monotone distribution on $G$, it is suffices to test if $p'$ is monotone or $\epsilon/4$-far from any monotone distribution on the bipartite poset $G'$.
\end{proof}

An interesting byproduct of Equation \ref{eq:violating_matching} is the following: If you consider the violation of each edge from monotonicity to be the weight of that edge, then the weight of the maximum weighted matching is the distance of the distribution to monotonicity. We formally explained it in the following theorem.


\distToMonMatching

\begin{proof} 
Let $W$ indicates the weight of the maximum weighted matching.
Fix a matching $M$ of $k$ edges $(u_i, v_i)$. Assume $p'$ is the closest monotone distribution to $p$, so $p'(u_i) \leq p'(v_i)$ for every edge $(u_i,v_i)$. 
One can show the following: 
\begin{align*}
d_{TV}(\Mon(G), p) & = \frac 1 2 \cdot \|p - p'\|_1 =  \frac 1 2\sum\limits_{(u_i,v_i) \in M} |p(u_i) - p'(u_i)| + |p(v_i) - p'(v_i)| 
\\
& \geq  \frac 1 2 \sum\limits_{(u_i,v_i) \in M} \max\left(0, p(u_i) - p(v_i) + p'(v_i) - p'(u_i)\right)
\\& \geq  \frac 1 2 \sum\limits_{(u_i,v_i) \in M} \max\left(0, p(u_i) - p(v_i)\right) \geq \frac 1 2 \, W
\end{align*}
where the last inequality is true, because the above is true for any matching $M$. On the other hand by Lemma \ref{lem:violating_matching}, there exists a (directed) matching $M_0$ in $TC(G)$, such that 
$$d_{TV}(\Mon(G), p) \leq \sum_{(u_i, v_i) \in M^*} p(u_i) - p(v_i) \leq W\, .$$
Thus, the proof is complete.
\end{proof}

\subsection{Proof of auxiliary lemmas}

\begin{lemma} \label{lem:violating_matching}
Let $p$ be a probability distribution over the vertex set $V$ of an unweighted directed graph $G = (V, E)$ representing a poset. Then, there exists a matching $M$ on the transitive closure $TC(G)$ such that 

$$\sum_{(u, v) \in M} p(u) - p(v) \geq d_{TV}(p, \Mon(G))\, .$$
\end{lemma}

\begin{proof} Define $\epsilon$ to be the $\ell_1$-distance of $p$ to monotonicity. We need to show the following: 
$$\sum_{(u, v) \in M} p(u) - p(v) \geq  \epsilon/2 \, .$$

Let $f^*$ be the monotone {\em function} on $G$ closest to $p$ (in the $\ell_1$-distance). Let $d$ denote $\|f^* - p\|_1$: the $\ell_1$-distance between $f^*$ and $p$. Note that $f^*$ is not necessarily a probability distribution which implies that $d$ can be smaller than $\epsilon$. 
To prove the above inequality, we will use $d$ as an intermediate variable which is in between the left hand side and the right hand side of the above inequality. Specifically, it suffices to prove the following:

\begin{enumerate}[(i)]\bfseries
\item $d \geq \epsilon/2$; \label{item:d_eps}
\item \textmd{ there exists a matching $M$ on the transitive closure of $G$ such that $\sum_{(u, v) \in M} p(u) - p(v)=d$.} \label{item:d_Matching}
\end{enumerate}

\noindent
\textbf{Proof of Item (\ref{item:d_eps}):} 
To show that $d$ is at least $\epsilon/2$, we prove that the monotone distribution $p_{f^*}$, obtained by normalizing $f^*$, is at most $2d$-far from $p$. Since any monotone distribution is at least $\epsilon$-far from $p$ in $\ell_1$-distance , we will have $\epsilon \leq \|p - p_{f^*}\|_1 \leq 2d$, establishing the desired claim.

First, note that if $f^*(v)$ is zero for all $v$, then by definition $d$ is at least $\epsilon/2$:
\begin{align*}
d = \sum\limits_{v \in V} \vert p(v) - f^*(v) \vert = \sum\limits_{v \in V} \vert p(v)\vert = 1 \geq \epsilon/2
\end{align*}
where the inequality holds since the $\ell_1$-distance between two distributions is always at most $2$, so $\epsilon$ is as well. Hence, assume $f^*$ is not a zero function for the rest of the proof. 

Also, note that $f^*$ is a non-negative function. We prove the non-negativity of $f^*$ by contradiction: assume $f^*(v)$ is negative for some $v$. Consider a non-negative function $f(v) = \max\{f^*(v), 0\}$. It is not hard to see that $f$ is monotone due to monotonicity of $f^*$. For every $v$ for which $f^*(v)<0$, we have
$$
\begin{array}{ll}
|p(v) - f(v)|  = p(v) - 0 < p(v) - f^*(v) = |p(v) - f^*(v)| \; . 
\end{array}
$$
Since $f^*(v) = f(v)$ everywhere else, $\|p - f\|_1 = \sum_{v\in V} |p(v) - f(v)| < \sum_{v\in V} |p(v) - f^*(v)| = \|p - {f^*}\|_1$ when $f^*$ contains some negative entry.
This contradicts the fact that $f^*$ was the closest monotone function to $p$, hence $f^*(v)$ has to be non-negative for all $v$'s.


Consider $p_{f^*}(v) = f^*(v)/ \sum_{u} f^*(u)$; it follows that $p_{f^*}$ is a well-defined monotone distribution. Then, 
$$
\begin{array}{lll}
\epsilon \leq \|p - p_{f^*}\|_1 & \leq & \|p - {f^*}\|_1 + \|f^*- p_{f^*}\|_1
= d + \sum\limits_{v \in V}{\left|{f^*}(v) - \dfrac{{f^*}(v)}{\sum_{u \in V} {f^*}(u)}\right|} 
\vspace{2mm}\\  
& = &d + \sum\limits_{v \in V}f^*(v) \cdot {\left| \dfrac{\left( \sum_{u \in V} f^*(u) \right)- 1}{\sum_{u \in V} f^*(u)}\right|}
=  d + \left| \sum_{u \in V} f^*(u) - 1\right| 
\vspace{2mm}\\  
& = & d + \left| \sum_{u \in V} f^*(u) - \sum_{u \in V} p(u) \right| 
\leq d + \sum_{u \in V} \left|  f^*(u) -  p(u) \right|
\vspace{2mm}\\  
& = & d + \|p - {f^*}\|_1 
= 2d 
\;.
\end{array}
$$
Thus, Item (\ref{item:d_eps}) is proved.

\noindent
\textbf{Proof of Item (\ref{item:d_Matching}):}  We leverage the duality theorem in linear programming. We write an LP that optimizes over all monotone functions $f$'s to find the function $f^*$ closest to $p$ under the $\ell_1$-distance. Let $x(v)$ be the variable that indicates the amount of perturbation at vertex $x$ that is needed to make $p$ monotone. For an edge $(u, v)$, the monotonicity constraint requires that 
$f(v) = p(v) + x(v) $ is at least $f(u) = p(u) + x(u)$, or equivalently,
$$x(v) - x(u) \geq p(u) - p(v) \; .$$
Given this inequality, we can find the monotone function closest to $p$ by solving the following linear program: 


$$
\begin{array}{llll}
\LP[3]: & \mbox{min} & \sum\limits_{v \in V}  |x(v)|
\\
& s.t. & x(v) - x(u) \geq p(u) - p(v) & \forall (u, v) \in E
\end{array}
$$
We denote the optimal solution for \LP[3] by $x^*(v) \coloneqq f^*(v) - p(v)$, and the corresponding optimal value of the objective function by $d \coloneqq \|p - {f^*}\|_1$.

To obtain the dual of \LP[3], we write down its standard form by substituting $x(v)$ by $x^+(v) - x^-(v)$ as follows:
$$
\begin{array}{llll}
\LP[4]: & \mbox{min}  &\sum\limits_{v \in V}  x^+(v) + x^-(v)
\vspace{2mm}\\
& s.t. & \left(x^+(v) - x^-(v)\right) - \left(x^+(u) - x^-(u)\right) \geq p(u) - p(v) & \forall (u, v) \in E
\vspace{2mm}\\
& & x^+(v), x^-(v) \geq 0 & \forall v \in V \;.
\end{array}  
$$
Then \LP[4] has the following dual:
$$
\begin{array}{llll}
\LP[5]: & \mbox{max} & \sum\limits_{(u, v) \in E}  \left(p(u) - p(v)\right) \cdot y(u, v)
\vspace{2mm}\\
& s.t. & \sum\limits_{(u, v) \in E} y(u, v) - \sum\limits_{(v, u) \in E} y(v, u)  \leq 1& \forall v \in V 
\vspace{2mm}\\
& & \sum\limits_{(v, u) \in E} y(v, u) - \sum\limits_{(u, v) \in E} y(u, v)   \leq 1& \forall v \in V
\vspace{2mm}\\
& & y(u, v) \geq 0 & \forall (u, v) \in E \; .
\end{array}  
$$
By strong duality, the optimal value of \LP[5] is equal to the optimal value of \LP[3], namely $d$.  On the other hand, the optimal solution of \LP[5] can help us to find a matching that satisfies the property in Item \ref{item:d_Matching}. Constraints of \LP[5] can be viewed in the form of $Ay \leq b$ and $y \geq 0$. Since $A$ is a {\em totally unimodular matrix} by Lemma \ref{lem:unimodularty_A} (proved below), the LP admits an optimal solution that is also \emph{integral}.

Let $y^*$ denote an integral optimal solution of the \LP[5], and let $S$ be a multi-set of the edges, containing $y^*(u, v)$ copies of edge $(u, v)$. Define the weight of each edge $(u, v)$ as $w(u,v)\coloneqq p(u) - p(v)$, and let the weight of a set $S$ be the sum of the weight of the edges in $S$. Thus:
$$w(S) \coloneqq \sum\limits_{(u, v) \in S} w(u, v)  = \sum\limits_{(u, v) \in S} p(u) - p(v) = \sum\limits_{(u, v) \in E} \left( p(u) - p(v) \right) \cdot y^*(u, v) = d\; .$$
We construct a matching $M$ where $w(M) = w(S)$, which completes the proof of Item \ref{item:d_Matching}. Based on the constraints of the $\LP[5]$, $S$ forms a subgraph on $G$ (but plausibly with multi-edges) such that the absolute difference between the number of incoming edges and outgoing edges at each vertex is at most one. Hence, we can decompose $S$ to paths and cycles.

Consider a path $P=\langle v_1, v_2, \ldots, v_k\rangle$. Observe that the weight of a path only depends on its endpoints:
$$w(P) = \sum\limits_{i=1}^{k-1} w(v_i, v_{i+1}) = \sum\limits_{i=1}^{k-1} p(v_i) -  p(v_{i+1}) = p(v_1) - p(v_k) = w(v_1, v_k)\;.$$
Remark that the edge $(v_1, v_k)$ does not necessarily belong to $E$, but since $v_1$ and $v_k$ are endpoints of a path $P$, then $(v_1, v_k)$ is contained in the \emph{transitive closure} of $G$.

By the above equation, if we replace the edges of $P$ in $S$ by a single edge $(v_1, v_k)$, then $w(S)$ remains unchanged. We can also remove all cycles without changing $w(S)$ since the weight of a cycle is always zero. Lastly, we may also join paths so that their endpoints are all distinct (since the difference between the in-degree and the out-degree of any vertex is at most one). After this process, we eventually obtain a matching $M$ on the transitive closure of $G$ such that 
 $$w(M) = \sum\limits_{(u, v) \in M} w(u, v) = w(S) = d\;,$$ 
concluding the proof of Item (\ref{item:d_Matching}) and this lemma. 
\end{proof}

\begin{lemma} \label{lem:unimodularty_A}
The matrix $A$, namely the coefficient matrix of \LP[5] when the constraints are written in the form $Ay\leq b$ and $y \geq 0$, is a totally unimodular matrix.
\end{lemma}
\begin{proof}
We arrange the rows of $A$ so that the two constraints of each vertex $v_i$ occupy two consecutive rows $2i-1$ and $2i$ for $i = 1, \ldots, n$, and that each column $j$ corresponds to the edge $e_j = (u_j, u'_j)$ for $j=1, \ldots, |E|$. 
Then, each entry of $A$ can be described as follows: 
$$
A_{i,j} = \left\{
\begin{array}{ll}\vspace{2mm}
1 & \quad \quad  \left(i \equiv 0 \ (\mathrm{mod}\ 2) \mbox{ and } u_j = v_{i/2} \right) \mbox{ or } 
\, \left(i \equiv 1 \ (\mathrm{mod}\ 2) \mbox{ and } u'_j = v_{(i+1)/2} \right)
\\ \vspace{2mm}
-1& \quad \quad  \left(i \equiv 1 \ (\mathrm{mod}\ 2) \mbox{ and } u_j = v_{(i+1)/2} \right) \mbox{ or } 
\, \left(i \equiv 0 \ (\mathrm{mod}\ 2) \mbox{ and } u'_j = v_{i/2} \right)
\\ \vspace{2mm}
0 & \quad \quad \mbox{otherwise} \,.\vspace{-2mm}
\end{array}
\right. 
$$
To prove that $A$ is a totally unimodular matrix, we make use of the following theorem.
\begin{theorem}[Ghouila-Houri Characterization \cite{Ghouila-Houri62}] An integral $m\times n$ matrix $A$ is a totally unimodular matrix if and only if, for any non-empty subset of rows, namely $R$, there exists a disjoint partition of $R$ into $R_1$ and $R_2$, such that the following is true. 
\begin{equation}\label{eq:ghouila}
\sum_{i \in R_1} A_{i,j}- \sum_{i \in R_2} A_{i,j} \in \{0, 1, -1\} \quad\mbox{  for  }\; j = 1, 2, \ldots, n \, .
\end{equation}
\end{theorem}
Here, for each non-empty subset $R \subseteq [2n]$, we explicitly define $R_1$ and $R_2$ according to the following three conditions. (1) If both $2i-1$ and $2i$ are in $R$, put both of them in $R_1$. 
(2) If only $2i-1$ is in $R$, then put $2i-1$ in $R_1$. (3) If only $2i$ is in $R$, then put $2i$ in $R_2$.

Consider column $j$ corresponding to $e_j = (v_r, v_{r'})$. This column has four non-zero entries:
$$
A_{2r-1, j} = -1, \quad A_{2r, j} = 1, \quad A_{2r'-1, j} = 1, \quad A_{2r', j} = -1 \,.
$$
If both $2r-1$ and $2r$ appear in $R$, or both of them are not in $R$, clearly Equation \ref{eq:ghouila} holds (similarly for $2r'-1$ and $2r'$). Thus, assume that exactly one of two rows $2r-1$ and $2r$, and exactly one of the two rows $2r'-1$ and $2r'$, are in $R$. It is not hard to see that if the corresponding entries $A_{i,j}$'s in these rows have the same sign, then one row ends up in $R_1$ and the other row ends up in $R_2$. If the entries have different signs, then both rows end up in the same set $R_1$ or $R_2$. In both of these cases, the sum in Equation \ref{eq:ghouila} becomes zero. Hence, the proof is complete. 
\end{proof}

%% file: bigness_UB.tex
\subsection{An Algorithm for Bigness Testing} \label{sec:bigness-ub}

We give an algorithm for the bigness testing problem that requires a sublinear number of samples. For testing bigness, all the domain elements must be at least a threshold $T$. The high level idea is to learn the {\em histogram} of the distribution
use a result from \cite{ValiantV17}. Then given the histogram, if the weight of the elements that are below the threshold is less than $\Theta(\epsilon)$, then we can accept the distribution, otherwise we reject. 

First, we define the histogram of a distribution. 
\begin{defn}
For a distribution $p$, we define $h_p:(0,1]\rightarrow \mathbb{N} \cup \{0\}$ to be the histogram of $p$ if and only if for all $x \in (0,1)$, $h(x)$ is the number of domain element $i$ such that $p(i)$ is equal to $x$. 
\end{defn}
Let $\pi:[n] \rightarrow [n]$ be a permutation of the domain elements. We define $p^{(\pi)}$ to be the permutation of $p$ according to $\pi$ such that for all domain element $i$, $p^{(\pi)}(i)$ is equal to $p(\pi(i))$. Based on the definition, it is not hard to see  permutation does not change the number of domain element with a certain probability, so $h(p)$ and $h(p^{(\pi)})$ are the same. Hence, when we learn the histogram of $p$,  we can claim that we learn $p$ {\em up to a permutation}. 

For learning, we will use a result from \cite{ValiantV17} for learning discrete distributions, up to a permutation of the domain elements. In Theorem 1.11 of \cite{ValiantV17}, combined with Fact 1 of \cite{ValiantV16}, authors provided the following theorem:

\begin{theorem}[\cite{ValiantV17, ValiantV16}]\label{thm:vv-ub}
There exists an algorithm that, given $O\left(\frac{n}{\epsilon^2 \log n} \right)$ i.i.d.~samples from an unknown distribution $p$, outputs an explicit description of a distribution, namely $q$, such that there exists a permutation $\pi:[n]\rightarrow[n]$ where $\sum_{i\in [n]} |p(i) - q(\pi(i))| \leq \epsilon$ with success probability $2/3$.
\end{theorem}

This theorem implies the following upper bound for bigness testing.

\begin{algorithm}[t]\label{alg:bigness}
\caption{Algorithm for Bigness Testing.}
\begin{algorithmic}[1]
\Procedure{Bigness-Test}{$\epsilon$, sample access to $p$}
 	\State{$\epsilon' \gets \epsilon/3$}
 	\State{$\mathcal{S} \gets$ Draw $O(\frac{n}{\epsilon'^2 \log n})$ samples from $p$}
    \State{$q\gets$ Learn $p$ (up to a permutation over $[n]$) via Theorem~\ref{thm:vv-ub} with error parameter $\epsilon'$ using samples in $\mathcal{S}$}
    \If{$d_{TV}(q,\Big(n,T)) \leq \epsilon'$}
		{\textbf{Return} \accept} 
    \Else 
    	{\textbf{Return} \reject}
    \EndIf
\EndProcedure
\end{algorithmic}
\end{algorithm}

\begin{corollary}
For bigness threshold $T \leq 1/n$, there exists an algorithm that distinguishes whether a distribution $p$ is $T$-big or $\epsilon$-far from $T$-big with success probability $2/3$ using $O(\frac{n}{\epsilon^2 \log n})$ i.i.d.~samples from $p$.
\end{corollary}

\begin{proof}
We refer to Algorithm~\ref{alg:bigness} for the outline of our procedure.
Let $q$ denote the distribution outputted by the ``learner" as promised by Theorem~\ref{thm:vv-ub} with distance parameter $\epsilon' = \epsilon/3$. Let $\pi$ be the permutation guaranteed by Theorem~\ref{thm:vv-ub}. We define $q'$ be the distribution obtained by permuting the elements of $q$ according to the associated permutation such that for each domain element $i$, let $q'(i) = q(\pi(i))$. Hence, with probability at least 2/3, $d_{TV}(p,q')$ is at most $\leq \epsilon'$. 
Note that $\pi$ is not known to the algorithm, but used for the analysis.
\\
Now, we have the following two cases: If $p$ is $T$-big, then $$d_{TV}(q',\Big(n,T)) \leq d_{TV}(q',p) \leq \epsilon' = \epsilon/3.$$ On the other hand, if $p$ is $\epsilon$-far from $T$-big, then $$d_{TV}(q',\Big(n,T)) \geq d_{TV}(p,\Big(n,T))-d_{TV}(p, q') \geq \epsilon - \epsilon' \geq 2\epsilon/3.$$ That is, $q$ offers us a condition for $T$-bigness testing by simply measuring its distance to $T$-bigness (the \textbf{if} condition of Algorithm~\ref{alg:bigness}). Therefore, Algorithm~\ref{alg:bigness} outputs the correct answer with probability at least 2/3. Note that learning $p$ using parameter $\epsilon' = \Theta(\epsilon)$ does not change the asymptotic sample complexity, so the proof is complete. 
\end{proof}

%% file: matchings_UB.tex
\subsection{An Algorithm for Testing Monotonicity on Matchings} \label{sec:matching-ub}
We  give a sublinear time algorithm for testing monotonicity on matchings. 
Similar to the previous section, we use a result from \cite{ValiantV17} for learning  the {\em distribution histogram} of a {\em pair of distributions}. First we employ the following definitions (see also Definition 5.2 and Definition 5.4 of \cite{ValiantV17}). A {\em distribution histogram} of a pair of distributions is a function that counts the number of elements with a given probability mass $x$ in the distribution $p_1$ and $y$ in the distribution $p_2$. More formally, we have the following definition: 

\begin{defn}[\cite{ValiantV17}]
For a pair of distributions $p_1$ and $p_2$, we say $h_{p_1,p_2}:[0, 1]^2\setminus\{(0,0)\}\rightarrow\mathbb{N}\cup\{0\}$ is the {\em distribution histogram} of $p_1$ and $p_2$ if and only if for any  $(x,y)$ in the domain:  $h_{p_1,p_2}(x,y) = |\{a: p_1(a) = x, p_2(a) = y\}|$.
\end{defn}

We will use this two-dimensional histogram to indicate a histogram of a distribution over a matching of size $n$: Let $p_1$ and $p_2$ be the two distributions that $p$ imposes on the  top and the bottom vertices in the matching respectively. Without loss of generality assume the edges in the matching connects the $i$-th vertex in the bottom to the $i$-th vertex in the top. Note that $h_{p_1,p_2}(x,y)$ counts the number of domain elements $a \in [n]$ such that  $p_1(a) = x$ and $p_2(a) = y$.  Hence, $\int_{x=0}^1 \int_{y=0}^1 h_{p_1,p_2}(x,y) dy \, dx$ is the number of matched pairs of vertices with at least one non-zero probability vertex. Since the sum of probabilities according to $p_1$ is one, we have $\int_{x=0}^1 \int_{y=0}^1 x \cdot h(x,y) = 1$. This is similarly true for $p_2$: $\int_{x=0}^1 \int_{y=0}^1 y \cdot h(x,y) = 1$. 
\\
Now, we define the distance between two histograms of two distributions: $h$ and $g$. 
At a high level, the distance between two histograms is the minimum cost one needs to pay to ``transform" $h$ to $g$.
In particular, we transform one histogram to another by moving mass from one point to another: By moving mass $c$ from $(x,y)$ to $(x',y')$, 
we obtain another histogram $h'$, such that $h'(x,y) = h(x,y) - c$, $h'(x',y') = h(x,y) + c$ and for all
other  points in $[0,1]^2$, $h$ and $h'$ are equal. 
The cost of this move is $c\cdot(|x - x'| + |y - y'|)$. More formally, we have the following definition.

\begin{defn} [\cite{ValiantV17}] \label{def:w-dist}
For a pair of functions $h, g:[0,1]^2 \setminus \{(0,0)\} \rightarrow \mathbb{N} \cup \{0\}$, we define the {\em distance notation $W(h,g)$} as the minimum cost over all mass moving schemes with finitely many steps for turning $h$ into $g$, where the cost for moving value $c > 0$ from point $(x,y)$ to $(x',y')$ is $c(|x-x'|+|y-y'|)$. Note that we assume that $\sum_{x,y} h(x,y) = \sum_{x,y} g(x,y)$, where extra value at point $(0,0)$ on $h$ or $g$ may be added to ensure this equality.
\end{defn}


Let $p^{(\pi)}$ be the {\em permuted distribution} of $p$ according to the permutation $\pi$ of $[n]$ such that for each domain element $i$, ${p'}_1^{(\pi)}(i) = p'_1(\pi(i))$. Note that as long as we permute $p_1$ and $p_2$ with the {\em same} permutation, the distribution histogram $h_{p_1, p_2}$ and $h_{p_1^{(\pi)}, p_2^{(\pi)}}$ are the same. Moreover, given $h_{p_1, p_2}$ one can construct $q_1$ and $q_2$ 
such that there exists a permutation $\pi$ for which $q_1$ and $q_2$ are the permuted versions of $p_1$ and $p_2$ according to $\pi$. 
\\
We relate the distance $W$ to the total variation distance in the following Lemma. In particular, the distance $W$ between two distribution histograms $h_{p_1,p_2}$, $h_{p'_1,p'_2}$ defined according to two pairs of distributions $(p_1, p_2)$, $(p'_1, p'_2)$ upper bounds the $\ell_1$-distance up to a permutation of the labels of the domain elements.
\\
\begin{lemma}\label{lem:WgeqTV}
 Let functions $h_{p_1,p_2}$, $h_{p'_1,p'_2}$ be defined according to two pairs of probability vectors $(p_1, p_2)$, $(p'_1, p'_2)$. There exists a permutation $\pi$ of $[n]$ such that 
 $$W(h_{p_1,p_2}, h_{p'_1,p'_2}) \geq \|p_1 - p'^{(\pi)}_1\|_1 + \|p_2 - p'^{(\pi)}_2\|_1.$$
\end{lemma}

\begin{proof}
According to the definition of the distance, $W$, there exists a moving scheme consisting of a sequence of $R$ steps, denoted by $\langle (c_r, (x_r, y_r), (x'_r, y'_r))\rangle_{r\in [R]}$ (with $c_r > 0$), describing the changes that eventually turn $h_{p_1,p_2}$ into $h_{p'_1,p'_2}$ for which we move the mass of $c_r$ from the source $(x_r, y_r)$ to sink $(x'_r, y'_r)$ at step $r$. 
We claim that if the scheme has minimum cost, $W(h_{p_1,p_2}, h_{p'_1,p'_2})$, without loss of generally, we may make the following assumptions about the scheme: (1)  There are no two steps $r_1$ and $r_2$ such that $(x'_{r_1},y'_{r_1})$ is the same as $(x_{r_2},y_{r_2})$.
(2) All the $c_r$'s are positive integers.
\\
To see why (1) is true, assume otherwise; if $r_1 = r_2$, then $(x'_{r_1},y'_{r_1}) = (x_{r_2},y_{r_2})$ means that the source and the sink in step $r_1$ is the same, so no mass is actually moved. Hence, we can just remove this step without changing the scheme. if $r_1 \neq r_2$, then $(x'_{r_1},y'_{r_1}) = (x_{r_2},y_{r_2})$ means that mass of quantity $\min(c_{r_1}, c_{r_2})$ is first moved from $(x_{r_1},y_{r_1})$ to $(x'_{r_1},y'_{r_1})$, and then moved from $(x'_{r_1},y'_{r_1})$ to $(x'_{r_2},y'_{r_2})$. Clearly, one can move the same quantity of mass from $(x_{r_1},y_{r_1})$ to $(x'_{r_2},y'_{r_2})$ directly with no larger cost, making one of the steps $r_1$ or $r_2$ vacuous. 
\\
Given (1), we now show that (2) also holds: Note that given (1), each point $(x,y)$ may appear in the steps as either a source or a sink, but not both. Moreover, the order of the steps does not matter, since the source always has the capacity for providing the mass. If there are several steps that move mass between the same source and the same sink, one can replace all of them with one step moving the total quantity of mass moved between them. Now, we can assume between each source and each sink there is a well defined quantity indicating how much mass we moved from the source to sink. This fact helps us to form a graph where the vertices are the sources and the sinks which appeared in the scheme. We put a directed edge from a source to a sink if we moved a non-integer mass from the source to the sink. We assign a weight to the edge which is the fractional part of the mass we moved from the source to the sink. We propose the following process for changing the steps for which each change removes at least one edge from the graph. We keep repeating the process  until no edge remains to assure that all $c_r$'s are integers. 
\\
Remove sources or sinks with no edge. Clearly, the graph is bipartite, and all the edges are from sources to sinks.  
Since $h_{p_1,p_2}$ and $h_{p'_1, p'_2}$ are integer, the final mass at each source and sink will eventually be an integer. Hence, each source has an out-degree of at least two and each sink has an in-degree of at least two. Therefore, the graph has an undirected cycle with an even length. Let $S$ and $T$ be the sets of the sources and the sinks involved in the cycle respectively. Let $E_1$ and $E_2$ be a partition of the edges in the cycle such that every other edge is in the same set. Clearly, each source (and sink) has exactly one edge in $E_1$ and one edge in $E_2$. As we define before the cost of moving one unit of mass via an edge from $(x,y)$ to $(x',y')$ is $|x - x'| + |y - y'|$. We define the cost of $E_1$ (and $E_2$) to be the total cost of edges in $E_1$ (and $E_2$). Without loss of generality assume cost of $E_1$ is not greater than the cost of $E_2$. Let $c^*$ be the minimum weight of edges in $E_2$. We modify the steps such that each step with a corresponding edge is $E_2$ moves $c^*$ less mass, and each steps with a corresponding edge in $E_2$ moves $c^*$ more mass. Clearly, this process does not increase the total cost of the scheme. However, it makes the fractional part of at least one step equal to zero. We repeat this process until no such step exists which concludes the proof for claiming (2).

Let $h^{(0)}, h^{(1)}, \ldots, h^{(R)}$ be the series of the distribution histograms which is generated during the mass moving scheme after each step. $h^{(0)}$ is the distribution histogram we start with, $h_{p_1, p_2}$, and $h^{(R)}$ is the final distribution histogram $h_{p'_1, p'_2}$.
Now, we create a sequence of pairs of \emph{vectors} $p^{(r)}_1, p^{(r)}_2 : [n] \rightarrow [0, 1]$ such that $h^{(r)} = h_{p^{(r)}_1,p^{(r)}_2}$ (under the same definition of distribution histogram, relaxed to allow non-distributions $p^{(r)}_1, p^{(r)}_2$). We start off with $p^{(0)}_1$ and $p^{(0)}_2$ being $p_1$ and $p_2$. Given $p^{(r-1)}_1, p^{(r-1)}_2$, we obtain $p^{(r)}_1,p^{(r)}_2$ as follows. 

Consider step $r$ described as $(c_r, (x_r,y_r), (x'_r,y'_r))$ with an integer $c_r$. Inductively, assume $h^{(r-1)} = h_{p^{(r-1)}_1,p^{(r-1)}_2}$ which implies that $p^{(r-1)}_1$ and $p^{(r-1)}_2$ contain at least $c_r \leq h^{(r-1)}(x_r, y_r)$ entries $i$ with $p^{(r-1)}_1(i)=x_r$ and $p^{(r-1)}_2(i)=y_r$. To apply step $r$, we pick an arbitrary set $I_r$ of $c_r$ many such entries, then modify the entries $p^{(r-1)}_1(i)$ and $p^{(r-1)}_2(i)$ from $x_r$ and $y_r$ to $x'_r$ and $y'_r$ respectively for each $i \in I_r$. That is, $p^{(r)}_1(i)=x'_r$ and $p^{(r)}_2(i)=y'_r$ for $i \in I_r$, and $p^{(r)}_1(i)=p^{(r-1)}_1(i)$ and $p^{(r)}_2(i)=p^{(r-1)}_2(i)$ for $i \notin I_r$. Hence, the \emph{$\ell_1$-distance} incurred by step $r$ becomes: 
\begin{align*}\|p^{(r-1)}_1-p^r_1\|_1 + \|p^{(r-1)}_2 - p^r_2\|_1 & = \sum_{i\in[n]}|p^{(r-1)}_1(i)-p^{(r)}_1(i)| + \sum_{i\in[n]}|p^{(r-1)}_2(i)-p^{(r)}_2(i)|
\\ & = \sum_{i\in I_r}|p^{(r-1)}_1(i)-p^{(r)}_1(i)| + \sum_{i\in I_r}|p^{(r-1)}_2(i)-p^{(r)}_2(i)|
\\ & = c_r |x_r-x'_r| + c_r |y_r-y'_r|\,.
\end{align*}
By summing over all $R$ steps, and applying the triangle inequality, we have:
\begin{align*} \|p_1 - p^R_1\|_1 + \|p_2 - p^R_2\|_1 &\leq \sum_{r\in [R]} \|p^{(r-1)}_1-p^r_1\|_1 + \sum_{r\in [R]} \|p^{(r-1)}_2-p^r_2\|_1 \\&= \sum_{r\in[R]} c_r |x_r-x'_r| + c_r |y_r-y'_r| 
\\ & = W(h^{(0)}, h^{(R)}) = W(h_{p_1,p_2},h_{p^{(R)}_1,p^{(R)}_2})\,.
\end{align*}

Now it remains to show that there exists a permutation $\pi$ that maps the labels of the given distribution $p'_1, p'_2$ to our constructed vectors $p^{(R)}_1, p^{(R)}_2$; namely, $p'^{(\pi)}_1 = p^{(R)}_1$ and $p'^{(\pi)}_2 = p^{(R)}_2$. Indeed, $h_{p'_1,p'_2}$ is the distribution histogram that counts the number of indices $i$ with $p'_1(i)=x$ and $p'_2(i)=y$, so $h_{p'_1,p'_2} = h_{p^R_1,p^R_2}$ implies that for every $(x,y)$, there are also equally many indices $i'$ with $p^R_1(i')=x$ and $p^R_2(i')=y$. Hence, there exists a bijection between their indices that maps $i'$'s to $i$'s and vice versa, concluding the lemma.
\end{proof}

Next, we state the the result of \cite{ValiantV17} to learn the distribution histogram of a pair of distributions. 

\begin{theorem}[Theorem 5.6 of \cite{ValiantV17}]\label{thm:vv-pair-ub}
There exists an algorithm that, given $O\left(\frac{n}{\epsilon^2 \log n} \right)$ i.i.d.~samples each from a pair of unknown distributions $p_1$ and $p_2$, outputs a function $g$ such that $W(h_{p_1,p_2},g) \leq \epsilon$ with success probability $2/3$.
\end{theorem}







\begin{algorithm}[t]\label{alg:matching-poset}
\caption{Algorithm for Testing Monotonicity over a Matching poset.}
\begin{algorithmic}[1]
\Procedure{Matching-Tester}{$\epsilon$, sample access to $p$}
    \State{$s \gets O\left(\frac{n}{\epsilon'^2 \log n}\right)$}
    \State{Comment: {generating $s$ samples from $p'$ (half $p$ and half uniform)}}
    \State{$\mathcal{S}\gets\emptyset$}
	\For{$i = 1, \ldots, s$}
    	\If{a (fresh) fair coin-toss comes up head}
    	    \State{Draw a sample from $p$ and add to $\mathcal{S}$}
    	\Else
    	    \State{Draw a uniform random vertex $x_i$ and add to $\mathcal{S}$ (where $x \in \{u,v\}$ and $i \in [n]$)}
        \EndIf
    \EndFor
	\State{$\epsilon' \gets \epsilon/14$}
	\State{$\widetilde{g}\gets$ Apply Theorem~\ref{thm:vv-pair-ub} for $p'$ with error parameter $\epsilon'$ using samples in $\mathcal{S}$}
	\State {$\hat{w}_S \gets $ Approximate total probability mass that $p'$ places on $S$ using $O(1/\epsilon')$ samples}
	\State {$\hat{w}_T \gets 1-\hat{w}_S$}
	\State {$\hat{g} \gets$ Rescale $\widetilde{g}$ to satisfy $\hat g(\hat w_S\cdot x , \hat w_T \cdot y)= \widetilde{g}(x,y)$}
	\State{$g^* \gets$ Compute a function minimizing $W(\hat g,g^*)$ defined according to a monotone distribution $q^*$}
	\If{$W(\hat{g},g^*) \leq 3\epsilon'$}
		\State{\textbf{Return} {\accept.}}
    \Else
    	\State{\textbf{Return} {\reject.}}
	\EndIf
\EndProcedure
\end{algorithmic}
\end{algorithm}

We now prove the upper bound for the monotonicity testing problem over the matching poset.

\begin{theorem}\label{thm:matchings}
For sufficiently small positive constant $\epsilon$, there exists an algorithm that distinguishes whether a distribution $p$ over the vertex set $V=S\cup T$ of a directed matching $M_n$ on $2n$ vertices is monotone or $\epsilon$-far from monotone with success probability $2/3$ using $O(\frac{n}{\epsilon^2 \log n})$ i.i.d.~samples from $p$.
\end{theorem}

\begin{proof}
For clarity, denote the edge set of the graph $G = (V, E)$ with the set of edges $E=\{(u_i, v_i)\}_{i\in [n]}$, and the set of vertices $V = S \cup T$ where $S=\{u_i\}_{i\in [n]}$ and $T=\{v_i\}_{i\in[n]}$. For a distribution $p$ over $V = S\cup T$, let $p_S$ and $p_T$ denote the probability mass $p$ places on elements of $S$ and $T$; note that $p_S$ and $p_T$ are functions on domain $S$ and $T$, but generally not probability distributions.

The outline of our algorithm is given as Procedure $\textsc{Monotonicity-Testing-over-}M_n$ in Algorithm~\ref{alg:matching-poset}. In our algorithm, we hope to invoke Theorem~\ref{thm:vv-pair-ub} by considering the (normalized) $p_S$ and $p_T$ as our $p_1$ and $p_2$, respectively. However, Theorem~\ref{thm:vv-pair-ub} requires roughly the same number of samples from both $p_1$ and $p_2$, while $p_S$ and $p_T$ may have vastly different total probability masses; for instance, it may be costly to try to obtain many samples from $S$.

Before we proceed, by Theorem \ref{thm:dist_to_mon_matching}, it is straightforward to see:
$$\sum_{i\in [n]} \max\{p(u_i)-p(v_i), 0\} \geq d_{TV}(p, \Mon(M_n)) \geq \frac 1 2 \sum_{i\in [n]} \max\{p(u_i)-p(v_i), 0\}.$$


In order to make the probability of the top and the bottom vertices at least a constant, we define an auxiliary probability distribution $p'$ obtained by averaging $p$ with a monotone distribution: $p'(w) = p(w)/2 + 1/(4n)$ where $w \in V$. Clearly, if $p$ is monotone, then $p'$ is monotone too. Also, if $p$ is $\epsilon$-far from monotone, then observe that the distance of $p'$ to monotone is 
\begin{align*}
d_{TV}(p', \Mon(M_n)) 
&\geq \frac 1 2 \sum_{i\in [n]} \max\{p'(u_i)-p'(v_i), 0\} \\
&\geq \frac 1 2 \sum_{i\in [n]} \max\left\{\left(\frac{p(u_i)}{2}+\frac{1}{4n}\right)-\left(\frac{p(v_i)}{2}+\frac{1}{4n}\right), 0\right\}\\
&\geq \frac 1 2 \sum_{i\in [n]} \max\left\{\frac{p(u_i)-p(v_i)}{2}, 0\right\}\geq\frac{1}{4} d_{TV}(p, \Mon(M_n)) \geq \frac \epsilon 4 \, ,
\end{align*} 
which preserves the distance to monotone to a factor of $4$. We can generate samples for $p'$ using asymptotically the same number of samples from $p$: A sample from $p'$ is obtained by drawing a sample from $p$ or drawing a uniform random vertex with probability $1/2$ each (Procedure $\textsc{Sample-from-}p'$ in Algorithm~\ref{alg:matching-poset}); henceforth, we consider the problem of testing $p'$ for monotonicity with distance $\epsilon/4$ instead.

The main benefit for considering the monotonicity testing on $p'$ instead of $p$ is that the total amount of probability masses placed on $S$ and on $T$ are at least $1/4 = \Omega(1)$ each. Hence, it takes $\Theta(s)$ samples from $p$ according to the procedure above to obtain at least $s$ samples from each of $S$ and $T$ with good constant probability; that is, we can create our input for the algorithm in Theorem~\ref{thm:vv-pair-ub} using $\Theta(s)$ i.i.d.~samples from $p$.

Denote by $w_S, w_T$ the total probability masses that $p'$ places on $S$ and $T$, respectively.  Let $p'_S$ and $p'_T$ be the probability function that $p$ assigns to vertices of $S$ and $T$, respectively. Let $\widetilde{p'_S}$ and $\widetilde{p'_T}$ be the distributions over $S$ and $T$ that are obtained by normalizing $p'_S$ and $p'_T$ (separately). More precisely, we have
$$\widetilde{p'_S}(i) = \frac{p'_S(i)}{w_S} = \frac{p'(u_i)}{w_S},  \mbox{ and } \ \widetilde{p'_T}(i) = \frac{p'_T(i)}{w_T} = \frac{p'(v_i)}{w_T} \quad \mbox{ for } i \in [n] \, .$$
Let $\epsilon'=\Theta(\epsilon)$ (to be determined exactly later). Invoking Theorem~\ref{thm:vv-pair-ub} with this parameter, we obtain a function $\widetilde{g}$ where $W(h_{\widetilde{p'_T}, \widetilde{p'_S}}, \widetilde{g}) \leq \epsilon'$ using $O(\frac{n}{\epsilon^2 \log n})$ samples from $p$.

Next, we rescale each dimension of $\widetilde{g}$ back by $w_S$ and $w_T$, thereby obtaining our estimate of $h_{p'_S, p'_T}$.
If we knew $w_S$ and $w_T$ exactly, we would define $g(w_S\cdot x , w_T \cdot y)= \widetilde{g}(x,y)$, and we would have $W(h_{p'_S,p'_T}, g) \leq \epsilon'$. However, we can only estimate $w_S$ and $w_T$ up to an additive error $\epsilon'$ with high constant probability using $O(1/\epsilon^2)$ samples. To this end, let $\hat{w}_S$ be the estimate of $w_S$, and let $\hat{w}_T = 1 - \hat{w}_S$. We define $\hat g$ for which $\hat g(\hat w_S\cdot x , \hat w_T \cdot y)= \widetilde{g}(x,y)$. Below, we show that $\hat g$ is a good estimation of $h_{p'_S, p'_T}$. 

Recall that $W(h_{\widetilde{p'_S}, \widetilde{p'_T}}, \widetilde{g}) \leq \epsilon'$. By definition, there exists a minimum-cost sequence of steps $\langle(c_r, (x_r, y_r), (x'_r, y'_r))\rangle_{r \in [R]}$ for turning $\widetilde{g}$ to $h_{\widetilde{p'_T}, \widetilde{p'_S}}$:
\begin{align*}
W(h_{\widetilde{p'_T}, \widetilde{p'_S}}, \widetilde{g})  & = \sum\limits_{r \in [R]} c_r \left(|x_r - x'_r| + |y_r - y'_r|\right) \leq \epsilon.
\end{align*}
Observe that under the cost function in Definition~\ref{def:w-dist}, we may assume without loss of generality that there are no $r, r'$ such that $(x,y) = (x'_r, y'_r) = (x_{r'}, y_{r'})$ in the moving scheme. Namely, we can instead ``shortcut'' this scheme by moving the value $\min\{c_r,c_{r'}\}$ from $(x_r, y_r)$ to $(x'_{r'}, y'_{r'})$ directly without leaving any extra amount at $(x,y)$ (during step $r$) to pick up later (during step $r'$). In this moving scheme, the value of $h_{p^r_S,p^r_T}$ on any $(x,y)$ must be non-increasing or non-decreasing throughout the steps $r \in [R]$ (since values are only being moved \emph{in}, or only being moved \emph{out}, but not a mixture of both). In particular, this condition implies that the total value of $c_r$'s moving into $(x',y')$ never exceeds the value of $h_{\widetilde{p'_S}, \widetilde{p'_T}}(x',y')$. More formally,
$$
\sum\limits_{r \ \mbox{s.t.} \,x'_r = x', y'_r = y'} c_r \leq h_{\widetilde{p'_S}, \widetilde{p'_T}}(x',y').
$$

Now, we are ready to bound $W(\hat g, h_{p'_S, p'_T})$.
By definition, we have $h_{p'_S, p'_T}(w_S \cdot x, w_T \cdot y)=h_{\widetilde{p'_S}, \widetilde{p'_T}}$ and $\hat g(\hat w_S \cdot x, \hat w_T \cdot y)=\tilde{g}(x,y)$. Thus, any moving scheme that turns $\tilde{g}$ into $h_{\widetilde{p'_S}, \widetilde{p'_T}}$, will also turn $\hat g$ into $ h_{p'_S, p'_T}$. Hence, we can use the same sequence (up to scaling) for moving the mass from $h_{\widetilde{p'_S}, \widetilde{p'_T}}$ to $\tilde g$ to show a bound for $W(\hat g, h_{p'_S, p'_T})$: at step $r \in [R]$, we move the value $c_r$ from $g(\hat w_S \cdot x, \hat w_T \cdot y)$ to $h(w_S \cdot x', w_T \cdot y')$. We establish our bound as follows.
\begin{align*}
W(\hat g, h_{p'_S, p'_T}) & \leq \sum\limits_{r \in [R]} c_r \left(| \hat w_S \cdot x_r - w_S\cdot x'_r| + |\hat w_T \cdot y_r - w_T \cdot y'_r|\right)
\\ 
& = \sum\limits_{r \in [R]} c_r \left(| \hat w_S \cdot x_r - \hat w_S\cdot x'_r + \hat w_S\cdot x'_r - w_S\cdot x'_r| + |\hat w_T \cdot y_r - \hat w_T \cdot y'_r + \hat w_T \cdot y'_r - w_T \cdot y'_r|\right)
\\
& \leq \sum\limits_{r \in [R]} c_r \left(\hat w_S |x_r - x'_r| + w_T|y_r-y'_r| \right)  + \sum\limits_{r \in [R]} c_r  \left(|w_S - \hat w_S|\cdot x'_r + |w_T - \hat w_T| \cdot y'_r\right)
\\
& \leq\left(\sum\limits_{r \in [R]} c_r \left(|x_r - x'_r| +|y_r-y'_r| \right) \right) + \epsilon'\cdot \left( \sum\limits_{r \in [R]} c_r  \left( x'_r + y'_r\right)\right)
\\
& \leq W(\tilde{g}, h_{\widetilde{p'_S}, \widetilde{p'_T}}) +\epsilon' \cdot \left(\int_{x=0}^\infty \int_{y=0}^\infty h_{\widetilde{p'_S}, \widetilde{p'_T}}(x,y)\cdot  (x + y)\,  dy \, dx\right)
\\
& \leq \epsilon' + \epsilon' \cdot \left(\sum_{i} \widetilde{p'_S}(i) + \widetilde{p'_T}(i)\right) = 3 \, \epsilon'.
\end{align*}

Going back to our algorithm, we compute $g^*$: the function minimizing $W(\hat g,g^*)$ that is also defined according to an actual \emph{monotone} probability distribution $q^*$ over $V$.  Observe that if $p'$ is monotone, then $$W(\hat g,g^*) \leq W(\hat g,h_{p'_S,p'_T}) \leq 3\epsilon'$$ due to the optimality assumption above. On the other hand, if $p'$ is $\epsilon/4$-far from monotone, then by choosing $\epsilon' = \epsilon/14$,
\begin{align*}
W(\hat g,g^*) &\geq W(h_{p'_S,p'_T},g^*) - W(h_{p'_S,p'_T}, \hat g) 
\\&\geq \|p',q^{*(\pi)}\|_1 - W(h_{p'_S,p'_T},\hat g) 
\\ &= 2\,d_{TV}(p',q^{*(\pi)}) - W(h_{p'_S,p'_T},\hat g) 
\geq 2(\epsilon/4) - 3\epsilon' = 4\epsilon',
\end{align*}
for some permutation $\pi$ over $[n]$, where $q^{*(\pi)}(u_i) = q^*(u_{\pi(i)})$ and $q^{*(\pi)}(v_i) = q^*(v_{\pi(i)})$, making use of Lemma~\ref{lem:WgeqTV} above. Hence, $g$ provides us with a condition for testing monotonicity over the matching poset $M_n$, as desired.
\end{proof}

%% file: bipartite_UB.tex
\subsection{An Algorithm for Testing Monotonicity on Bounded Degree Bipartite Graphs with Sub-linear Sample Complexity} \label{sec:bipartite-ub}

We give an algorithm which tests monotonicity of a distribution $p$ on a \emph{bipartite} poset $G$ with sample complexity $O\left(\frac{\Delta^3 n}{\epsilon^2 \log n} \right)$ where $\Delta$ denotes an upper bound for the maximum degree over all vertices in $G$. Given sample access to the distribution $p$, we implement a sampling oracle for a certain distribution $p'$ on a \emph{matching} poset $G'$ with $O(\Delta n)$ vertices. This distribution $p'$ is monotone on $G^\prime$ if  $p$ is monotone on $G$, and $p'$ is $\epsilon/(2\Delta)$-far from monotone on $G'$ if $p$ is $\epsilon$-far on $G$. Hence, we apply the algorithm for testing monotonicity on the matching poset $G'$ to test the monotonicity of $p'$, immediately obtaining  the desired sample complexity. We describe the construction of $G'$ and the distribution $p'$ below and show the correctness of our approach in Theorem~\ref{thm:UB_boundedDeg}.

More formally, let $p$ be a distribution over a directed bipartite poset $G = (V=V_B\cup V_T, E)$ where $V_B = \{u_i\}_{i\in[n]}$ and $V_T=\{v_i\}_{i\in[n]}$  are the sets of the bottom and the top vertices, and $E \subseteq V_B \times V_T$ is the set of edges. Let $\Delta$ be an upper bound on the degree of $G$.

\paragraph{The matching poset $G'$.} Based on $G$,
we create a matching $G'= (V'=V'_b\cup V'_t,E')$ over $n' = O(\Delta n)$ vertices according the following procedure. 
Similar to $G$, $V'_b$ is the set of bottom vertices,  $V'_t$ is the set of top vertices, and $E'$ is the set of edges. 

\begin{compactitem}
\item Create $\Delta$ \emph{copy vertices} $w^1, \ldots, w^{\Delta}$ for each vertex $w \in V$.

\item For each edge $e = (u,v) \in E$, match an unmatched pair of vertices $u^i, v^j$ via the \emph{copy edge} $e' = (u^i, v^j)$; place $u^i \in V'_b$, $v^j \in V'_t$ and $e' \in E'$.
\item For all remaining unmatched vertices $w^i$, create a \emph{dummy vertex} $\bar{w}^i$, then match it to $w^i$ via the \emph{dummy edge} $\bar{e}_{w^i} = (\bar{w}^i,w^i)$; place $\bar{w}^i \in V'_b$, $w^i \in V_T'$ and $\bar{e}_{w^i} \in E'$. Note that the dummy vertex is always put in the bottom set.
\end{compactitem}
Note that the second step above is always possible since there are at most $\Delta$ edges incident to a vertex. 

\paragraph{Distribution $p'$ over $G'$.} The distribution $p'$ over the poset $G'$ is defined as follows. For each copy vertex $w^i$, set $p'(w^i) = p(w)/\Delta$. For each dummy vertex $\bar{w}^i$, set $p'(\bar{w}^i)=0$.
One can generate a sample from $p'$, by drawing a sample $w$ in $V$ according to $p$, and drawing $i$ uniformly at random from $[\Delta]$: The $i$-th copy of $w$, $w^i$, is a sample drawn from $p'$.

In the following lemma, we show that the distance of $p'$ to being monotone is closely related to the distance of $p$ to monotonicity. 
\begin{lemma} \label{lem:dist-p-p-prime}
Let $p$ and $p'$ be two distributions over $G$ and $G'$ as described above. If $p$ is monotone, then $p'$ is monotone. If $p$ is $\epsilon$-far from being monotone, then $p'$ is $(\epsilon/2\Delta)$-far from being monotone. 
\end{lemma}

\begin{proof}
Observe that for each copy edge $e' = (u^i, v^j)$, the probabilities at the endpoints are $p'(u^i) = p(u)/\Delta$ and $p'(v^j) = p(v)/\Delta$, respectively. Thus, if $p(u)$ is at most $p(v)$, then $p'(u^i)$ will remain at most $p'(v^j)$. Furthermore, for each dummy edge $\bar{e}_{w^i} = (\bar{w}^i,w^i)$, the probability of the bottom vertex, $p'(\bar{w}^i)$, is zero, so this edge never violates the monotonicity of $G'$. Hence it follows immediately that if $p$ is monotone on $G$, then $p'$ is monotone on $G'$ as well. 

On the other hand, assume $p$ is $\epsilon$-far from being monotone. We define a weighted graph on the transitive closure of $G$, $TC(G)$, where the weight of each edge $(u,v)$ is $\max(p(u) - p(v), 0)$. By the proof of Theorem~\ref{thm:dist_to_mon_matching}, $TC(G)$ has a weighted matching, namely $M$, of weight $W$ such that 
\begin{equation}\label{eq:dtvtoMon}
    \frac {W} 2 \leq d_{TV} (\Mon(G),p) \leq W \,.
\end{equation}

Since $G$ is a bipartite poset, and the edges are all from $V_B$ to $V_T$,  $TC(G)$ is the same as $G$. Hence, each edge $e = (u,v)$ in $M$ is in $E$ as well. Also, by the construction of $G'$, there exists a copy edge $e' = (u^i, v^j) \in E'$ that corresponds to $e$. Let $M'$ be the set of copy edge $e' = (u^i, v^j)$ where $e = (u,v)$ is in $M$.  $M'$ is a matching in $G'$ as well. 

Observe that by the above construction, the weight of $e' = (u^i, v^j)$ is $ \max(p'(u^i)-p'(v^j), 0) = \max(p(u)-p(v), 0)/\Delta$. 
Hence, $G'$ contains a matching, $M'$, of weight $W' \coloneqq W/\Delta$ which is at most the weight of the maximum matching in $G'$. Let $W'$ be the weight of the maximum matching in $G'$. By Theorem~\ref{thm:dist_to_mon_matching} and Equation~\ref{eq:dtvtoMon}, we obtain:
$$
 \frac{d_{TV} (\Mon(G),p)}{2 \Delta} \leq \frac{W}{2 \Delta} \leq \frac{W'}{2} \leq d_{TV}(\Mon(G'), p')\,.
$$
Thus, if $p$ is $\epsilon$-far from being monotone, then $p'$ is $\epsilon/(2\Delta)$-far from monotone as well, concluding the lemma. 
\end{proof}

Given the above lemma, it is sufficient to test monotonicity of $p'$ with proximity parameter $\epsilon' = \epsilon/(2\Delta)$. See Algorithm~\ref{alg:boundedDeg} for the steps. Below, we show the correctness of the algorithm.

\begin{algorithm}[t]\label{alg:boundedDeg}
\caption{Reduction from testing monotonicity on a bipartite poset to a matching poset.}
\begin{algorithmic}[1]
\Procedure{Reduction}{$G,  n, \Delta, \epsilon, $ sample access to $p$}
    \State {$\epsilon' \gets \epsilon/(2\Delta)$}
	\State {$G' \gets $Construct the matching poset from $G$ as described.}
	\State {$\mathcal{S} \gets$ Generate $O(\frac{\Delta^3 n}{\epsilon'^2 \log n})$ samples from $p'$}
	\State {Test if $p'$ is monotone or $\epsilon'$-far from it via Algorithm~\ref{alg:matching-poset} using the samples in $\mathcal{S}$.}
	\State{Output the result of the test.}
\EndProcedure
\end{algorithmic}
\end{algorithm}

\begin{corollary} \label{thm:UB_boundedDeg}
There exists an algorithm that tests whether a distribution $p$ over a bipartite poset $G$ of $n$ vertices and maximum degree $\Delta$, is monotone or $\epsilon$-far from monotone with success probability $2/3$, using $O(\frac{\Delta^3 n}{\epsilon^2 \log n})$ i.i.d.~samples from $p$.
\end{corollary} 
\begin{proof}
Given Lemma~\ref{lem:dist-p-p-prime}, it suffices to test the monotonicity of $G'$ with parameter $\epsilon' = \epsilon/\Delta$. Using Theorem~\ref{thm:matchings} and since $G'$ is a matching of size $n' = O(\Delta n)$, one can test monotonicity of $p'$ with high probability using $O(n'/(\epsilon'^2 \log n')) = O({\Delta^3 n}/({\epsilon^2 \log n}))$ samples as desired. Therefore, the proof is complete. 
\end{proof}